\documentclass[12pt]{amsart}

\usepackage{amsmath,amsthm,amsfonts,amssymb,eucal, bbm}
\usepackage{scalerel}
\usepackage{mathabx}

\usepackage{color}

\renewcommand {\a}{ \alpha }
\renewcommand{\b}{\beta}

\newcommand{\g}{\gamma}

\renewcommand{\d}{\delta}
\newcommand{\s}{\sigma}
\renewcommand{\l}{\lambda}

\newcommand{\z}{\zeta}
\renewcommand{\t}{\theta}

\newcommand{\p}{\partial}
\newcommand{\om}{\omega}
\newcommand{\Om}{\Omega}

\newcommand{\oq}{\ {\raise 7pt\hbox{${\scriptstyle\circ}$}}
	\kern -7pt{
		\hbox{$Q$}}}

\newcommand{\R}{ \mathbb R}

\newcommand {\BA}{\mathbf A}

\newcommand {\BS}{\mathbf S}

\newcommand {\bx}{\mathbf x}

\newcommand{\SfG}{{\sf{G}}}
\newcommand{\sg}{{\sf{g}}}

\newcommand{\SPi}{{\sf{\Pi}}}

\newcommand{\ST}{{\sf{T}}}
\newcommand{\SU}{{\sf{U}}}
\newcommand{\SV}{{\sf{V}}}
\newcommand{\sv}{{\sf{v}}}

\newcommand{\SfS}{{\sf S}}

\newcommand {\BOLG}{\boldsymbol{\sf\Gamma}}
\newcommand{\CK}{\mathcal K}

\newcommand{\CH}{\mathcal H}

\newcommand{\CM}{\mathcal M}

\newcommand{\CC}{\mathcal C}

\newcommand{\plainW}[1]{\textup{{\textsf{W}}}^{#1}}
\newcommand{\plainC}[1]{\textup{{\textsf{C}}}^{#1}}
\newcommand{\plainB}{\textup{{\textsf{B}}}}

\newcommand{\plainH}[1]{\textup{{\textsf{H}}}^{#1}}

\newcommand{\plainL}[1]{\textup{{\textsf{L}}}^{#1}}
\newcommand{\plainl}[1]{\textup{{\textsf{l}}}^{#1}}

\DeclareMathOperator{\tr}{{tr}}

\newcommand{\scale}{{\scaleto{1}{3pt}}}

\newcommand{\scalel}[1]
{{\scaleto{#1}{3pt}}}
\newcommand{\scalet}[1]
{{\scaleto{#1}{4pt}}}

\DeclareMathOperator{\iop}{{\sf Int}}

\DeclareMathOperator{\supp}{{supp}}

\DeclareMathOperator{\dc}{d}

\newtheorem{thm}{Theorem}[section]
\newtheorem{cor}[thm]{Corollary}
\newtheorem{lem}[thm]{Lemma}
\newtheorem{prop}[thm]{Proposition}

\theoremstyle{definition}

\newtheorem{rem}[thm]{Remark}

\numberwithin{equation}{section}

\newcommand{\bee}{\begin{equation}}
	\newcommand{\ene}{\end{equation}}
\newcommand{\bees}{\begin{equation*}}
	\newcommand{\enes}{\end{equation*}}
\newcommand{\bes}{\begin{split}}
	\newcommand{\ens}{\end{split}}

\newcommand{\bet}{\begin{thm}}
	\newcommand{\ent}{\end{thm}}
\newcommand{\bel}{\begin{lem}}
	\newcommand{\enl}{\end{lem}}
\newcommand{\bec}{\begin{cor}}
	\newcommand{\enc}{\end{cor}}
\newcommand{\bep}{\begin{proof}}
	\newcommand{\enp}{\end{proof}}
\newcommand{\ber}{\begin{rem}}
	\newcommand{\enr}{\end{rem}}

\newcommand{\Z}{\mathbb Z}

\newcommand{\1}{\mathbbm 1}

\begin{document}
	\hoffset -4pc

\title
[Eigenvalue asymptotics]
{{Eigenvalue asymptotics for the one-particle kinetic energy density operator}} 
\author{Alexander V. Sobolev}
\address{Department of Mathematics\\ University College London\\
	Gower Street\\ London\\ WC1E 6BT UK}
 \email{a.sobolev@ucl.ac.uk}

\begin{abstract} 
The kinetic energy of a multi-particle system is described by the one-particle 
kinetic energy density matrix $\tau(x, y)$. Alongside the 
one-particle density matrix $\gamma(x, y)$ it 
is one of the key objects in the quantum-mechanical approximation schemes. 
We prove the asymptotic formula $\l_k \sim (Bk)^{-2}$, $B \ge 0$, 
as $k\to\infty$,   
for the eigenvalues $\l_k$ of the self-adjoint operator 
$\boldsymbol\ST\ge 0$ with kernel $\tau(x, y)$. 
\end{abstract}

\maketitle

\section{Introduction} 
 
\subsection{Background} 
Consider on $\plainL2(\R^{3N})$ the Schr\"odinger operator 
 \begin{align}\label{eq:ham}
H = \sum_{k=1}^N \bigg(-\Delta_k -  \frac{Z}{|x_k|}
\bigg) 
 + \sum_{1\le j< k\le N} \frac{1}{|x_j-x_k|},
\end{align}
describing an atom with $N$ charged particles  
with coordinates $\bx = (x_1, x_2, \dots, x_N),\,x_k\in\R^3$, $k= 1, 2, \dots, N$, 
and a nucleus with charge $Z>0$. The notation $\Delta_k$ is used for 
the Laplacian w.r.t. the variable $x_k$. 
The operator $H$ acts on the Hilbert space $\plainL2(\R^{3N})$ 
and it is self-adjoint on the domain 
$D(H) =\plainW{2,2}(\R^{3N})$, since the potential in \eqref{eq:ham} 
is an infinitesimal perturbation 
relative to the unperturbed operator $-\Delta = - \sum_k \Delta_k$, 
see e.g. \cite[Theorem X.16]{ReedSimon2}.   
Let $\psi = \psi(\bx)$ 
be an eigenfunction of the operator $H$ with an eigenvalue $E\in\R$, i.e. $\psi\in D(H)$ and 
\begin{align*}
(H-E)\psi = 0.
\end{align*}  
For each $j=1, \dots, N$, we represent
\begin{align*}
\bx = (\hat\bx_j, x_j), \quad \textup{where}\ 
\hat\bx_j = (x_1, \dots, x_{j-1}, x_{j+1},\dots, x_N),
\end{align*}
with obvious modifications if $j=1$ or $j=N$. 
The \textit{one-particle density matrix} is defined as the function 
\begin{align}\label{eq:den}
\g(x, y) = \sum_{j=1}^N\int\limits_{\R^{3N-3}} \overline{\psi(\hat\bx_j, x)} \psi(\hat\bx_j, y)\  
d\hat\bx_j,\quad (x,y)\in\R^3\times\R^3. 
\end{align} 
This is a major object in quantum mechanics, see e.g. 
\cite{RDM2000}, \cite{Davidson1976}, \cite{LLS2019}, \cite{LiebSei2010}. 
Introduce also the function 
\begin{align}\label{eq:tau}
\tau(x, y) = 
\sum_{j=1}^N\int\limits_{\R^{3N-3}} \overline{\nabla_x\psi(\hat\bx_j, x)} 
\cdot\nabla_y \psi(\hat\bx_j, y)\  
d\hat\bx_j, 
\end{align}
that we call the \textit{one-particle kinetic energy density matrix}. 
The choice of this term does not seem to be standard, but it is partly 
motivated by the fact that the trace
\begin{align*}
\tr \boldsymbol \ST = \int_{\R^3} \tau(x, x)\, dx 
\end{align*}
gives the kinetic energy of the $N$ particles, see e.g. 
\cite[Section 2A]{Davidson1976}, \cite[Chapter 3]{LiebSei2010} or 
\cite[Section 4]{LLS2019}. 
Denote by $\BOLG = \iop(\g)$ and $\boldsymbol\ST = \iop(\tau)$ the integral operators 
with kernels $\g(x, y)$ and $\tau(x, y)$ respectively. 
We call $\BOLG$ \textit {the one-particle density operator}, 
and $\boldsymbol\ST$ - \textit {the one-particle kinetic energy density operator}. 
If we assume that the particles are spinless fermions (resp. bosons), i.e. that the eigenfunction 
is fully antisymmetric (resp. symmetric) with respect to the permutations $x_k\leftrightarrow x_j$, 
$j, k = 1, 2, \dots, N, j\not = k$, 
then the formulas for $\g(x, y)$ and $\tau(x, y)$ take a more compact form:
\begin{align}\label{eq:fb}
\begin{cases}
\g(x, y) = N \int\limits_{\R^{3N-3}} 
\overline{\psi(\hat\bx, x)} \psi(\hat\bx, y)\ d\hat\bx,\\[0.5cm]
\tau(x, y) = N\int\limits_{\R^{3N-3}} \overline{\nabla_x\psi(\hat\bx, x)} 
\cdot\nabla_y \psi(\hat\bx, y)\,d\hat\bx,
\end{cases}
\end{align}
where $\hat\bx = \hat\bx_N$. In this paper, however, 
we do not make these assumptions and 
work with the general definitions \eqref{eq:den} and \eqref{eq:tau}.
Since $\psi, \nabla\psi\in\plainL2(\R^{3N})$, 
both $\BOLG$ and $\boldsymbol\ST$ are trace class.
We are interested in the exact decay rate of the eigenvalues $\l_k, k= 1, 2, \dots,$ 
for $\BOLG$ and $\boldsymbol\ST$. 
Both operators play a central 
role in quantum chemistry computations of atomic and molecular bound states, see e.g.  the papers 
\cite{Fries2003_1, Fries2003, Lewin2004, Lewin2011}  and the book \cite{SzaboOstlund1996}. 
The knowledge of the eigenvalue behaviour should 
serve to estimate the errors due to finite-dimensional 
approximations,  see e.g. 
\cite{Cioslowski2020}, 
\cite{CioStras2021},  \cite{Fries2003} and \cite{HaKlKoTe2012}.

The eigenvalue asymptotics for the operator $\BOLG$ were studied in \cite{Sobolev2021} 
under the assumption that the function $\psi$ satisfies an exponential bound 
\begin{align}\label{eq:exp}
\sup_{\bx\in\R^{3N}} e^{\varkappa |\bx|_{\scalel{1}}}\, 
 |\psi(\bx)|<\infty
\end{align}  
with some $\varkappa>0$, where $|\,\cdot\,|_1$ denotes the $\plainl1$-norm in $\R^{3N}$. 
This condition  always holds if $E$ is a discrete eigenvalue. For detailed discussion 
and bibliography on the property \eqref{eq:exp} we refer to 
\cite{SimonSelecta}. 
It was shown in \cite{Sobolev2021} that 
\begin{align}\label{eq:bolg}
\lim_{k\to \infty} k^{\frac{8}{3}} \,\l_k(\BOLG) = A^{\frac{8}{3}}
\end{align}
with some coefficient $A\ge 0$, see also \cite{Cioslowski2020}  and \cite{CioPrat2019} 
for relevant quantum chemistry calculations. 
In the author's paper \cite{Sobolev_Nov2021} the main ideas of the proof 
were presented for the simplified variant of the density matrix, containing only one 
term from the sum on the right-hand side of \eqref{eq:den}. 

The subject of the current paper is to establish 
for the eigenvalues $\l_k(\boldsymbol\ST)$ the formula 
\begin{align}\label{eq:mainas}
\lim_{k\to\infty} k^{2}\l_k(\boldsymbol\ST) = B^2,
\end{align}
with a coefficient $B\ge 0$. The precise statement, including the formula for 
the coefficient $B$,  requires further properties 
of the function $\psi$ and 
it is given in Theorem \ref{thm:maincompl}. Moreover, in Sect. \ref{sect:reg} for comparison 
we also provide a formula for the coefficient $A$. 

\subsection{Outline of the proof}
The proof follows the plan of \cite{Sobolev2021}, and it 
splits into three steps that are briefly summarized below.  

\textit{Step 1: factorization.} 
First we represent the operator $\boldsymbol\ST$ as 
the product $\boldsymbol\ST = \boldsymbol\SV^*\boldsymbol\SV$, where 
$\boldsymbol\SV:\plainL2(\R^3)\mapsto \plainL2(\R^{3N-3}; \mathbb C^{3N})$ is a 
suitable compact operator. 
This representation implies that $\l_k(\boldsymbol\ST) = s_k(\boldsymbol\SV)^2$, $k=1, 2 \dots,$, 
where $s_k(\boldsymbol\SV)$ are the singular values of the operator $\boldsymbol\SV$. 
Therefore the asymptotic formula \eqref{eq:mainas} takes the form 
\begin{align}\label{eq:main1}
\lim_{k\to\infty}k\, s_k(\boldsymbol\SV) = B.
\end{align}
If we take formula \eqref{eq:fb} as the definition of $\tau(x, y)$ 
(which is the case for fermions or bosons), then the operator $\boldsymbol\SV$ can be taken to be 
the integral operator 
acting from $\plainL2(\R^3)$ into $\plainL2(\R^{3N-3}; \mathbb C^3)$ 
with the vector-valued kernel $\sqrt{N}\,\nabla_x\psi(\hat\bx, x)$, i.e. 
\begin{align}\label{eq:ferbos}
(\boldsymbol\SV u)(\hat\bx) = \sqrt{N}\int_{\R^3} 
\nabla_x\psi(\hat\bx, x) u(x) d x,\  u\in\plainL2(\R^3).
\end{align}
For the general case the operator $\boldsymbol\SV$ 
has a more complicated form and it is defined 
in Subsect. \ref{subsect:fact}. 
The factorization $\boldsymbol\ST = \boldsymbol\SV^*\boldsymbol\SV$ 
reduces the study of the asymptotics 
to the operator $\boldsymbol\SV$ which is convenient since the properties of 
the functions $\psi$ have been very well studied in the literature. 
For the rest of the introduction we will assume for simplicity that $\boldsymbol\SV$ is given by 
the formula \eqref{eq:ferbos}. 

For the sake of comparison we note that the operator $\BOLG$ also admits a factorization: $\BOLG = \boldsymbol\Psi^* \boldsymbol\Psi$. Under the assumption that 
$\g(x, y)$ is given by the formula \eqref{eq:fb} the 
scalar operator $\boldsymbol\Psi:\plainL2(\R^3)\mapsto \plainL2(\R^{3N-3})$ has the form 
\begin{align}\label{eq:gam}
(\boldsymbol\Psi u)(\hat\bx) = \sqrt{N}\int_{\R^3} 
\psi(\hat\bx, x) u(x) d x,\  u\in\plainL2(\R^3).
\end{align}
This factorization was a key observation in the study of the operator $\BOLG$ in \cite{Sobolev2021}.

\textit{Step 2: estimates for singular values.} 
The asymptotic analysis of the operator $\boldsymbol\SV$ begins with effective bounds for the 
singular values of $\boldsymbol\SV$. By effective bounds we mean bounds for the 
operator $\boldsymbol\SV$ with weights 
$a = a(x), x\in\R^3$, $b = b(\hat\bx), \hat\bx\in \R^{3N-3}$,  
of the form 
\begin{align*}
s_k(b\, \boldsymbol\SV a)\le C(a, b) k^{-1},
\end{align*}
where the factor on the right-hand side depends explicitly on some 
integral norms of $a$ and $b$.  
The precise statement is the subject of Theorem \ref{thm:psifull}.  
%
%
%
%
%
%
The appropriate bounds for the operator $\boldsymbol\Psi$ were derived in \cite{Sobolev2020} 
using the following two ingredients: 
global estimates for the derivatives of $\psi$ 
obtained by S. Fournais and T.\O. S\o rensen in the recent paper \cite{FS2021},  
and 
the results by 
M.S. Birman and M.Z. Solomyak \cite{BS1977} on 
estimates for singular values of  
integral operators via norms of their kernels in certain 
\textit{Sobolev} spaces. 
Later the author realized that Sobolev spaces are not the most suitable ones 
for these purposes. Indeed,   
application of the bounds from \cite{FS2021} 
%
%
%
in the Sobolev spaces framework makes the proof in \cite{Sobolev2020}
quite long and unnaturally complicated. 
In the current paper we observe that 
%
%
the Fournais-S\o rensen 
bounds almost immediately imply that the kernel $\nabla_x\psi(\hat\bx, x)$ 
%
%
belongs to a certain \textit{Besov} 
space. Thus a different group of estimates from \cite{BS1977} is applicable -- the one that uses 
the Besov spaces instead of the Sobolev spaces. This approach 
makes the proof of the effective bounds for $\boldsymbol\SV$ much shorter and more direct.  

%
%
%
%
%

The Fournais-S\o rensen paper \cite{FS2021} is one of the most recent contributions to the 
long history of regularity properties for the multi-particle Schr\"odinger equation. 
We refer to \cite{FS2021} for further references to the vast bibliography on this subject. 
%
%

%
%

\textit{Step 3: asymptotics.}  This step follows the approach of \cite{Sobolev2021} where the 
asymptotic formula \eqref{eq:bolg} was proved.  
In order to obtain an 
asymptotic formula for $\boldsymbol\SV$ 
one needs to know the behaviour of the kernel $\nabla_x\psi(\hat\bx, x)$ 
near the coalescence points. An appropriate representation 
for the function $\psi$ was obtained 
in \cite{FHOS2009}. 
To simplify the explanation we assume that the system consists of two particles, 
i.e. that $N=2$. 
Under this assumption the problem retains all its crucial features, 
but permits to avoid some tedious technical details. For $N=2$ we have $\bx = (t, x)\in \R^3\times\R^3$, 
and the operator $\boldsymbol\SV$ as defined in \eqref{eq:ferbos} 
acts from $\plainL2(\R^3)$ into $\plainL2(\R^3; \mathbb C^3)$. 
According to \cite{FHOS2009}, there exists a neighbourhood 
$\Om_{1,2}\subset \big(\R^3\setminus \{0\}\big)\times \big(\R^3\setminus\{0\}\big)$ of the 
diagonal set $\{(x, x): x\in\R^3\setminus \{0\}\}$ 
and two functions $\xi_{1,2}, \eta_{1, 2}$, real analytic in $\Om_{1,2}$,  such that 
the eigenfunction $\psi = \psi(t, x)$ admits the representation
\begin{align}\label{eq:locan0}
\psi(t, x) = \xi_{1,2}(t, x) + |t-x|\,\eta_{1,2}(t, x),
\quad \textup{for all}\quad (t, x)\in \Om_{1,2},
\end{align}
and hence the kernel of the operator $\boldsymbol\SV$ has the form $\sqrt 2\, \nabla_x\psi(t, x)$, 
where  
\begin{align}\label{eq:locan}
\nabla_x\psi(t, x) = &\ \nabla_x\xi_{1,2}(t, x) + |t-x|\,\nabla_x\eta_{1,2}(t, x) 
+ \frac{x-t}{|x-t|}\eta_{1,2}(t, x),\notag\\
&\ \qquad \qquad\textup{for all}\quad (t, x)\in \Om_{1,2}.
\end{align}
Each of the terms on the right-hand side 
of \eqref{eq:locan} 
gives a different contribution 
to the asymptotics of the singular values. According to \cite{BS1977}, infinitely smooth kernels 
lead to a decay rate which is faster than any negative power of the singular value number. 
This allows us to say that the first term on the right-hand side of \eqref{eq:locan} gives a zero contribution. 
The other two terms contain homogeneous factors of order one and zero respectively. 
Spectral asymptotics for operators with homogeneous kernels were studied extensively by 
M. Birman and M. Solomyak in \cite{BS1970, BS1977_1} and \cite{BS1979}, 
see also \cite{BS1977}. A summary of the results needed for our purposes here can be found 
in \cite{Sobolev2021}.

The homogeneity order one kernels played a central role in paper 
\cite{Sobolev2021} in the study of the operator $\boldsymbol\Psi$ defined in \eqref{eq:gam}. 
In this case the Birman-Solomyak theory led 
to the asymptotics 
\begin{align*}
\lim_{k\to\infty} k^{\frac{4}{3}} s_k(\boldsymbol\Psi) = A^{\frac{4}{3}},
\end{align*}
which, in its turn, implied \eqref{eq:bolg}. 
In the current paper the decay of the singular values of the operator 
$\boldsymbol\SV$ is determined by the term of homogeneity order zero in \eqref{eq:locan}. 
Precisely, according to the Birman-Solomyak results (see \cite{Sobolev2021} for a summary), 
kernels of order zero produce the asymptotics   
of order $k^{-1}$ , which agrees with \eqref{eq:main1}. However, as in \cite{Sobolev2021}, 
the known formulas for integral operators with homogeneous kernels are not directly applicable, 
since we have information neither on the smoothness of the functions $\xi_{1, 2}, \eta_{1, 2}$ 
on the closure  
$\overline{\Om_{1,2}}$ nor on the integrability of these functions and 
their derivatives on $\Om_{1, 2}$. 
To resolve this problem we 
approximate $\xi_{1,2}, \eta_{1,2}$ by suitable 
$\plainC\infty_0$-functions supported inside $\Om_{1,2}$. 
The estimates obtained in Step 2 ensure that the induced error in the spectral asymptotics 
tends to 
zero as the smooth approximations converge.  
In the limit this leads to the formula 
\eqref{eq:main1} with the constant
\begin{align*}
B = \frac{4}{3\pi}\int_{\R^3} |2^{1/2}\eta_{1,2}(x, x)| dx.
\end{align*}  
Interestingly, 
the finiteness of the above integral is a by-product of the 
spectral bounds. 
We emphasize that  the coalescence points $x = 0$ or $t = 0$ do not contribute to the asymptotics.  

For $N\ge 3$ the proof needs to be modified to incorporate representations of the form 
\eqref{eq:locan0} in the neighbourhood of all pair coalescence points 
$x_j = x_k$, \ $j, k = 1, 2, \dots, N$,\ $j \not = k$, see \eqref{eq:repr}. 
The kernel of the operator 
$\boldsymbol\SV$ near the 
coalescence points becomes more complex and as a result the proof requires an 
extra step involving an 
 integral operator whose structure mimics that of the operator $\boldsymbol\SV$. 
Such a model operator was studied in the author's paper \cite{Sobolev2021}. In order 
to make it adaptable for the use in different  situations, the model operator was allowed to have 
arbitrary order of homogeneity at the pair coalescence points. For order one the model operator 
was the key point in the analysis for  
the operator $\boldsymbol\Psi$ conducted in \cite{Sobolev2021}. In the current paper we 
use the same model operator but with homogeneity of order zero, which is relevant to the operator 
$\boldsymbol\SV$. The required results are collected in Subsect. \ref{subsect:model}. 
Using these facts together with 
approximations of $\boldsymbol\SV$ similar to the ones in the case $N=2$, 
we arrive at the asymptotics \eqref{eq:main1} with the coefficient $B$ given in \eqref{eq:coeffA}, 
thereby concluding the proof.  

\subsection{Plan of the paper}
The paper is organized as follows. 
Section \ref{sect:reg} contains the main result and some preliminaries. 
First we provide the necessary facts about regularity of the function $\psi$ including 
its representation near the pair coalescence points. This allows us to state   
the main result (Theorem \ref{thm:maincompl}) which includes 
formula \eqref{eq:coeffA} for the coefficient $B$. 
The important Theorem \ref{thm:etadiag} ensures that the coefficient $B$ is 
finite.  
In Subsect. \ref{subsect:compact} 
we put together some general information on  
compact operators with power-like spectral behaviour.
The focus of Sect. \ref{sect:intop} is on the basic spectral results for 
integral operators.  
In Subsect. \ref{subsect:nikol} 
we establish an elementary but very useful test of membership in some Besov spaces.  
It is used in Subsect. \ref{subsect:BS} to adapt 
spectral bounds from \cite{BS1977} for the use with the operator $\boldsymbol\SV$. 
Subsect. \ref{subsect:model} provides spectral asymptotics 
of the model integral operator that had been examined in 
\cite{Sobolev2021}. 
The proof of the main results begins in Sect. \ref{sect:factor}. First, in Subsect \ref{subsect:fact} 
we detail the 
factorization $\boldsymbol\ST = \boldsymbol\SV^*\boldsymbol\SV$ which allows us to recast 
Theorem \ref{thm:maincompl} in terms of the operator $\boldsymbol\SV$, see Theorem 
\ref{thm:ttov}. 
In Lemma \ref{lem:central} we construct an approximation for $\boldsymbol\SV$ 
which enables us to use the model operator from Subsect. \ref{subsect:model} 
to prove  
Theorem \ref{thm:ttov} (and hence Theorem \ref{thm:maincompl}) 
in Subsect. \ref{subsect:proofs}.  The last Section 
\ref{sect:trim} is focused on the proof 
of the approximation Lemma \ref{lem:central}. It begins with the derivation of 
effective spectral bounds for the operator $\boldsymbol\SV$ in Subsect. \ref{subsect:prep}. 
This step is immediate as it is directly based on the bounds derived in Subsect. \ref{subsect:BS}.
 The bounds for $\boldsymbol\SV$ are used in Subsect. \ref{subsect:trim} to show that  
the error in the asymptotics induced by the approximations of $\boldsymbol\SV$ tends to 
zero as the approximations converge, thereby proving Lemma \ref{lem:central}. 
In the Appendix we establish an elementary extension result for the Sobolev spaces which is used in 
Subsect. \ref{subsect:nikol}.

\subsection{Notational conventions} 
We conclude the introduction with some general notational conventions.  

\textit{Coordinates.} 
As mentioned earlier, we use the following standard notation for the coordinates: 
$\bx = (x_1, x_2, \dots, x_N)$,\ where $x_j\in \R^3$, $j = 1, 2, \dots, N$. 
The vector $\bx$ is often represented in the form  
\begin{align*}
\bx = (\hat\bx_j, x_j) \quad \textup{with}\quad   
\hat\bx_j = (x_1, x_2, \dots, x_{j-1}, x_{j+1},\dots, x_N)\in\R^{3N-3}, 
\end{align*}
for arbitrary $j = 1, 2, \dots, N$. Most frequently we use this notation with $j=N$, and  
write $\hat\bx = \hat\bx_N$, so that $\bx = (\hat\bx, x_N)$. 
In order to write 
formulas in a more compact and unified way, we sometimes use the notation 
$x_0 = 0$. 

In the space $\R^d, d\ge 1,$ 
the notation $|x|$ stands for the Euclidean norm. 
Sometimes it is convenient to use the $\plainl1$-norm of $x$ 
which we denote by $|x|_1$.

For $N\ge 3$ it is also useful to introduce the notation 
for $\bx$ with $x_j$ and $x_k$ taken out: 
\begin{align}\label{eq:xtilde}
\tilde\bx_{k, j} = \tilde\bx_{j, k} 
= (x_1, \dots, x_{j-1}, x_{j+1}, \dots, x_{k-1}, x_{k+1},\dots, x_N), \quad \textup{for}\  j <k.
\end{align}
If $j < k$, then we write $\bx = (\tilde\bx_{j, k}, x_j, x_k)$. For any $j\le N-1$ 
the vector $\hat\bx$ can be represented as $\hat\bx = (\tilde\bx_{j, N}, x_j)$.

The notation $B_R$ is used for the Cartesian product 
$\bigtimes_{j=1}^N \{x_j\in\R^3: |x_j| <R\}$. 

\textit{Derivatives.} 
Let $\mathbb N_0 = \mathbb N\cup\{0\}$.
If $x = (x', x'', x''')\in \R^3$ and $m = (m', m'', m''')\in \mathbb N_0^3$, then 
the derivative $\p_x^m$ is defined in the standard way:
\begin{align*}
\p_x^m = \p_{x'}^{m'}\p_{x''}^{m''}\p_{x'''}^{m'''}.
\end{align*} 
For any $l = 0, 1, 2, \dots$, we use the notation 
\begin{align}\label{eq:nabla}
|\nabla_x^l u| = \sum_{|m|=l}|\p_x^m u|.  
  \end{align}
The symbols $\plainW{l, p}(\Om)$, $\plainW{l, p}_{\textup{loc}}(\Om)$, $l = 0, 1, \dots$, 
$1\le p\le \infty$, 
 stand for the standard Sobolev spaces on the open set $\Om$.   
  
\textit{Bounds.} 
For two non-negative numbers (or functions) 
$X$ and $Y$ depending on some parameters, 
we write $X\lesssim Y$ (or $Y\gtrsim X$) if $X\le C Y$ with 
some positive constant $C$ independent of those parameters.
 To avoid confusion we may comment on the nature of 
(implicit) constants in the bounds. 
We also use the notation $X\wedge Y = \min\{X, Y\}$. 

\textit{Cut-off functions.} 
We systematically use the following smooth cut-off functions. Let 
\begin{align}\label{eq:sco}
\t, \z\in \plainC\infty([0, \infty)), 
\quad \z = 1-\t, 
\end{align}
be functions such that $0\le \t\le 1$ and 
\begin{align}\label{eq:sco1} 
\t(t) = 0,\quad \textup{if}\quad t>1;\ \quad
\t(t) = 1,\quad \textup{if}\quad 0\le t<\frac{1}{2}. \ 
\end{align}

\textit{Integral operators.} 
The notation $\iop(\CK)$ is used for the integral operator with kernel $\CK$, 
e.g. $\BOLG = \iop(\g)$.  
The functional spaces, where $\iop(\CK)$ acts are obvious from the context.

\section{Regularity of $\psi$, main result,  
compact operators}\label{sect:reg}
    
\subsection{Regularity of the eigenfunction} 
Our argument relies on the following regularity properties of the eigenfunction $\psi$. 
For convenience, along with the standard 
notation $\bx = (x_1, x_2, \dots, x_N)\in\R^{3N}$ 
below we use the notation $x_0 = 0$. 
Thus, unless otherwise stated, the indices labeling the particles, run from $0$ to $N$. 
Denote the set 
\begin{align}\label{eq:sls}
\SfS_{l,s} = \SfS_{s, l} = \{\bx\in\R^{3N}: x_l\not = x_s\},\ 
l\not = s. 
\end{align}
Since the Coulomb potential is real analytic away from the origin, 
by an elementary ellipticity argument (see e.g. \cite{Hor1976}), 
the eigenfunction $\psi$ is real analytic away from the coalescence points, i.e. on the set 
\begin{align*}
\SU = 
\bigcap_{0\le l < s\le N}\,\SfS_{l,s}.
\end{align*}
Apart from this qualitative fact we need explicit bounds for $\psi$ and its derivatives 
w.r.t. $x_j$, $j = 1, 2, \dots, N,$ on the set  
\begin{align*}
\SPi_j = \bigcap_{\substack{l\not= j\\
0\le l\le N}}\, \SfS_{l, j},
\end{align*}
which does not contain the points of coalescence of the variable $x_j$ with the remaining ones, 
but allows $x_k = x_s$ for all $k\not= j, s\not \not = j$. 
Such bounds were obtained 
by S. Fournais and T.\O. S\o rensen in \cite{FS2021}. Let 
\begin{align*}
\dc(\hat\bx_j, x_j) = \min \{1, |x_j-x_k|, \ 0\le k \le N, \, k\not = j\},\quad  j = 1, 2, \dots, N.
\end{align*}
The following proposition is a consequence of \cite[Corollary 1.2]{FS2021}. 
Recall that $|\, \cdot\, |_{\scale}$ denotes the $\plainl1$-norm.

\begin{prop}\label{prop:FS} 
Assume that $\psi$ satisfies 
\eqref{eq:exp}. Then for all $\bx = (\hat\bx_j, x_j)\in\SPi_j$, $j = 1, 2, \dots, N$, 
we have the bound 
\begin{align}\label{eq:FS}
|\nabla_{x_j}^m\psi(\hat\bx_j, x_j)|\lesssim \big(1+\dc(\hat\bx_j, x_j)^{1-m}\big)e^{-\varkappa {|\bx|_\scale}}, 
\, \quad \textup{for all}\quad m = 0, 1, \dots,
\end{align}
where we have used the notation \eqref{eq:nabla}.
\end{prop}
  
In order to state the main result we also need the formula for the function $\psi$ describing its behaviour 
in a neighbourhood of the pair coalescence points, obtained in \cite{FHOS2009}.  
Along with the set $\SU$, for each pair $j, k: j\not = k$, define the set 
\begin{align}\label{eq:uj}
\SU_{j,k} =  
\bigcap_{\substack{l \not = s\\
(l, s)\not = (j, k)}}  
\SfS_{l,s}.
\end{align}
It contains the coalescence point $x_j = x_k$, 
whereas the other pairs of variables are separated from each other. 
We are interested in the shape of function $\psi$ 
near the diagonal
\begin{align}\label{eq:diag}
\SU^{(\rm d)}_{j,k} = \{\bx\in\SU_{j,k}: x_j = x_k\}.
\end{align}
The sets introduced above are obviously symmetric with respect to permutations of 
indices, e.g. $\SU_{j,k} = \SU_{k,j}$, $\SU^{(\rm d)}_{j, k}=\SU^{(\rm d)}_{k,j}$.

The following property follows from \cite[Theorem 1.4]{FHOS2009}.

\begin{prop}\label{prop:repr}
For each pair of indices $j, k = 0, 1, \dots, N$ such that $j\not = k$,  there exists 
an open connected set $\Om_{j,k} = \Om_{k, j}\subset \R^{3N}$, such that 
\begin{align}\label{eq:omin}
\SU_{j,k}^{(d)}\subset\Om_{j,k}\subset \SU_{j,k},
\end{align}
and two uniquely defined functions $\xi_{j, k}, \eta_{j, k}$, both 
real analytic on $\Om_{j, k}$, such that for all $\bx\in \Om_{j, k}$  
the following representation holds:
\begin{align}\label{eq:repr}
 \psi(\bx) = \xi_{j, k}(\bx) + |x_j-x_k| \eta_{j, k}(\bx).
\end{align} 
\end{prop} 

Because of the uniqueness of functions $\xi_{j, k}, \eta_{j, k}$, they are 
symmetric with respect to the permutations $j\leftrightarrow k$, i.e. 
$\xi_{j, k} = \xi_{k, j}$, $\eta_{j, k} = \eta_{k, j}$ for all $j\not = k$. 
  
The asymptotic coefficient $B$ in the formula \eqref{eq:mainas} is defined via the functions 
$\eta_{j,k}$, $j, k = 1, 2, \dots, N, j<k$, on the sets \eqref{eq:diag}.  
Using the notation \eqref{eq:xtilde} we write the function $\eta_{j, k}(\bx)$ 
on $\SU_{j, k}^{(\dc)}$ as $\eta_{j, k}(\tilde\bx_{j, k}, x, x)$. Proposition \ref{prop:repr} 
gives no information on the properties of $\xi_{j, k}, \eta_{j, k}$ on the closure $\overline{\Om_{j, k}}$, 
but while proving the asymptotic formula \eqref{eq:mainas} we find the following integrability 
property of $\eta_{j, k}(\tilde\bx_{j,k}, x, x)$.

\begin{thm}\label{thm:etadiag} 
If $N\ge 3$, then each function $\eta_{j, k}(\ \cdot\ , x, x)$, $1\le j < k\le N$, 
belongs to $\plainL2(\R^{3N-6})$ a.e. $x\in \R^3$ and   
the function 
\begin{align}\label{eq:H}
H(x):= \bigg[2 \sum\limits_{1\le j < k\le N}\int_{\R^{3N-6}} \big| 
\eta_{j, k}(\tilde\bx_{j,k}, x, x) \big|^2 d\tilde\bx_{j, k}\bigg]^{\frac{1}{2}},
\end{align} 
 belongs to $\plainL1(\R^{3})$. 
 
If $N = 2$, then the function $H(x):= \sqrt2 |\eta_{1, 2}(x, x)|$ 
belongs to $\plainL1(\R^{3})$.
\end{thm}

Now we are in a position to state the main result of the paper.
 
\begin{thm}\label{thm:maincompl}
Suppose that the eigenfunction $\psi$ satisfies the bound \eqref{eq:exp}. 
Then the eigenvalues $\l_k(\boldsymbol\ST), k = 1, 2, \dots,$ 
of the operator $\boldsymbol\ST$  satisfy the asymptotic formula \eqref{eq:mainas} 
with the constant 
\begin{align}\label{eq:coeffA}
B = \frac{4}{3\pi}
 \int_{\R^{3}} H(x)\, dx. 
\end{align} 
\end{thm}
  
A few remarks are in order.   
%
%
\begin{enumerate}
\item 
Theorem \ref{thm:maincompl} extends to the case of a 
molecule with several nuclei whose positions are fixed. 
The modifications are straightforward.
\item 
The formula \eqref{eq:coeffA} shows that the spectral asymptotics 
depend only on the behaviour of the eigenfunction $\psi$ near the pair coalescence points 
$x_j=x_k$, $j, k = 1, 2, \dots, N$, $j\not = k$.  
Neither the points  
$x_j = 0, j = 1, 2, \dots, N$, nor 
the coalescence points of higher orders (e.g. $x_j = x_k = x_l$ with pair-wise distinct $j, k, l$)   
contribute to the asymptotics \eqref{eq:mainas}. 
\item 
For the sake of comparison we give here the formula for the coefficient 
$A$ in the asymptotic formula \eqref{eq:bolg} for the operator $\BOLG$. 
As shown in \cite[Theorem 2.2]{Sobolev2021}, the function $H$ also belongs to $\plainL{3/4}(\R^3)$ 
and 
\begin{align*}
A = \frac{1}{3}\bigg(\frac{2}{\pi}\bigg)^{\frac{5}{4}}
 \int_{\R^{3}} H(x)^{\frac{3}{4}}\, dx. 
\end{align*} 
Thus both asymptotic formulas \eqref{eq:bolg} and \eqref{eq:mainas} are determined only by 
the functions $\eta_{j, k}$.
\item 
Both coefficients $A$ and $B$ may be equal to zero. 
Indeed, if the particles 
are spinless fermions, i.e. the eigenfunction $\psi$ 
is fully antisymmetric, 
then so are the functions $\xi_{j, k}$ and $\eta_{j, k}$ in 
\eqref{eq:repr}, and hence $\eta_{j, k}(\tilde\bx_{j, k}, x, x) = 0$ for 
all $\bx\in\SU_{j, k}^{(\dc)}$, which means that $H(x) = 0$, a.e. $x\in\R^3$. 

If we introduce the spin, then the parity of the functions $\xi_{j, k}$ and $\eta_{j, k}$ depends 
on the configuration formed by the pair $j, k$, see \cite[Subsect. 3.3]{HaKlKoTe2012}. If $j$ and $k$ 
are in a \textit{triplet} configuration, then 
$\xi_{j, k}$ and $\eta_{j, k}$ are antisymmetric with respect to the permutation $j\leftrightarrow k$  and hence $\eta_{j, k}(\tilde\bx_{j, k}, x, x) = 0$, as before. 
If however, $j$ and $k$ form a \textit{singlet}, 
then $\xi_{j, k}$ and $\eta_{j, k}$ are 
symmetric with respect to the permutation $j\leftrightarrow k$ and 
$\eta_{j, k}(\tilde\bx_{j, k}, x, x)$ is not identically zero on $\SU_{j, k}^{(\dc)}$, 
so both $A$ and $B$ are positive. 
A somewhat more detailed discussion can be found in \cite[Remark 2.4]{Sobolev2021}.

For completeness one should say that a preliminary analysis shows that 
in the totally antisymmetric case (i.e. when all pairs are triplets and hence $A=B=0$) 
the eigenvalues of $\BOLG$ and 
$\boldsymbol\ST$ should satisfy the bounds 
$\l_k(\BOLG) = O(k^{-10/3})$ and $\l_k(\boldsymbol\ST) = O(k^{-8/3})$. Proof of these 
bounds will make part of another publication.  

\end{enumerate}
 
\subsection{Compact operators} \label{subsect:compact} 
Here we provide necessary information about compact operators. 
Most of the listed facts can be found in \cite[Chapter 11]{BS}. 
Let $\CH$ and $\mathcal G$ be separable Hilbert spaces.  
Let $A:\CH\to\mathcal G$ be a compact operator. 
If $\CH = \mathcal G$ and $A=A^*\ge 0$, then $\l_k(A)$, $k= 1, 2, \dots$, 
denote the positive eigenvalues of $A$ 
numbered in descending order counting multiplicity. 
For arbitrary spaces $\CH$, $\mathcal G$ and compact $A$, by $s_k(A) >0$, 
$k= 1, 2, \dots$, we denote the singular values of 
$A$ defined by $s_k(A)^2 = \l_k(A^*A) = \l_k(AA^*)$.   
If $s_k(A)\lesssim k^{-1/p}, k = 1, 2, \dots$, 
with some $p >0$, then we say that $A\in \BS_{p, \infty}$ and denote
\begin{align*}
\| A\|_{p, \infty} = \sup_k s_k(A) k^{\frac{1}{p}},
\end{align*}
so that $\|A^*A\|_{p/2, \infty} = \|A A^*\|_{p/2, \infty} = \|A\|_{p, \infty}^2$.
The class $\BS_{p, \infty}$ is a complete linear space with the quasi-norm $\|A\|_{p, \infty}$. 
For all $p>0$ the functional $\|A\|_{p, \infty}$ 
satisfies the following ``triangle" inequality for  
operators $A_1, A_2\in\BS_{p, \infty}$:
\begin{align}\label{eq:triangle}
\|A_1+A_2\|_{\scalet{p}, \scalel{\infty}}^{{\frac{\scalel{p}}{\scalet{p+1}}}}
\le \|A_1\|_{\scalet{p}, \scalel{\infty}}^{{\frac{\scalel{p}}{\scalet{p+1}}}}
+ \|A_2\|_{\scalet{p, \infty}}^{{\frac{\scalel{p}}{\scalet{p+1}}}}.
\end{align}  
This inequality allows us to estimate quasi-norms of ``block-vector" operators. 
Let $A_j\in\BS_{p, \infty}$ 
be a finite collection of compact operators. Define the operator $\BA:\CH\to\oplus_j \mathcal G$ by  
$\BA = \{A_j\}_{j}$.
Since 
\begin{align*}
\BA^* \BA = \sum_j A_j^*A_j,  
\end{align*}
and $A_j^*A_j\in\BS_{q, \infty}$, $q = p/2$, 
by \eqref{eq:triangle} we have 
\begin{align}\label{eq:blockvec}
\|\BA\|_{p, \infty}^{{\frac{\scalel{2p}}{\scalet{p+2}}}} 
= \|\BA^*\BA\|_{q, \infty}^{{\frac{\scalel{q}}{\scalet{q+1}}}}
\le  \sum_j\|A_j^*A_j\|_{q, \infty}^{{\frac{\scalel{q}}{\scalet{q+1}}}} = 
\sum_j\|A_j\|_{p, \infty}^{{\frac{\scalel{2p}}{\scalet{p+2}}}}.
\end{align}
Consequently, in order to estimate the singular values of $\BA$ it suffices to 
estimate those of its components $A_j$. We use this fact throughout the paper. 
 
For $p\in (0, 1)$ it is often more convenient to use the following ``triangle" inequality for  
arbitrarily many operators $A_j\in\BS_{p, \infty}$, $j = 1, 2, \dots$:
\begin{align}\label{eq:ptriangle}
\big\|\sum_{j} A_j\big\|_{p, \infty}^p\le (1-p)^{-1}
\sum_j \|A_j\|_{p, \infty}^p,
\end{align}
see 
\cite[Lemmata 7.5, 7.6]{AJPR2002}, \cite[\S 1]{BS1977} and references therein. Under 
the additional assumption
\begin{align}\label{eq:orthog}
A_k A_j^* = 0\quad \textup{or}\quad A_k^*A_j = 0,\quad j\not = k, 
\end{align}   
the inequality of the form \eqref{eq:ptriangle} extends to all $p\in (0, 2)$: 

\begin{prop}\cite[Lemma 1.1]{BS1977}\label{prop:orthog}
Let $p\in (0, 2)$ and let $A_j\in \BS_{p, \infty}$, $j= 1, 2, \dots$, 
be a family of operators satisfying 
\eqref{eq:orthog}. 
Then for the operator $A = \sum_j A_j$ we have the inequality
\begin{align}\label{eq:pnorm}
\|A\|_{\scalet{p}, \scalel{\infty}}^{\scalet{p}}\le 
\frac{2}{2-p}\,
\sum_j \|A_j\|_{\scalet{p}, \scalel{\infty}}^{\scalet{p}}. 
\end{align} 
\end{prop}

\begin{proof} 
It suffices to conduct the proof under the first condition in 
\eqref{eq:orthog} only, so that  
\begin{align*}
A A^* = \sum_j A_j A_j^*. 
\end{align*}
Since $A_j A_j^*\in\BS_{q, \infty}, q = p/2<1$, 
and $\|A_j A_j^* \|_{q, \infty} = \|A_j\|_{p, \infty}^2$, 
the inequality \eqref{eq:ptriangle} leads to the bound 
\begin{align*}
\|A\|_{p, \infty}^p = \|A A^*\|_{q, \infty}^q \le (1-q)^{-1} 
\sum_j \|A_j A_j^*\|_{q, \infty}^q
= (1-p/2)^{-1}\sum_j \|A_j\|_{p, \infty}^p. 
\end{align*}
This inequality coincides with \eqref{eq:pnorm}.
\end{proof}

Let us now introduce the functionals describing the asymptotic behaviour of singular values. 
For $A\in \BS_{p, \infty}$ the following quantities are finite:
\begin{align}\label{eq:limsupinf}
\begin{cases}
\SfG_p(A) = 
\big(\limsup\limits_{k\to\infty} k^{\frac{\scalel{1}}{\scalel{p}}}s_k(A)\big)^{p},\\[0.3cm]
\sg_p(A) =  
\big(\liminf\limits_{k\to\infty} k^{\frac{\scalel{1}}{\scalel{p}}}s_k(A)\big)^{p},
\end{cases}
\end{align}
and they clearly satisfy the inequalities
\begin{align} \label{eq:gnorm}
\sg_p(A)\le \SfG_p(A)\le \|A\|_{p, \infty}^p.
\end{align}
Note that if 
$A\in\BS_{p, \infty}$, then 
$\SfG_q(A) = 0$ for all $q > p$. 
Observe that 
\begin{align}\label{eq:double}
\sg_{p/2}(A A^*) = \sg_{p/2}(A^* A) = \sg_{p}(A),\quad 
\SfG_{p/2}(A A^*) = \SfG_{p/2}(A^* A) = \SfG_{p}(A).
\end{align}
If $\SfG_p(A) = \sg_p(A)$, then the singular values of $A$ satisfy the asymptotic formula
\begin{align*} 
s_k(A) = \big(\SfG_p(A)\big)^{\frac{1}{p}} k^{-\frac{1}{p}} + o(k^{-\frac{1}{p}}),\ k\to\infty.
\end{align*}
The functionals $\sg_p(A)$, $\SfG_p(A)$ 
also satisfy the inequalities of the type \eqref{eq:triangle}:
\begin{align}\label{eq:trianglep}
\begin{cases}
\SfG_p(A_1 + A_2)^{{\frac{\scalel{1}}{\scalet{p+1}}}}
\le \SfG_p(A_1)^{{\frac{\scalel{1}}{\scalet{p+1}}}}
+ \SfG_p(A_2)^{{\frac{\scalel{1}}{\scalet{p+1}}}},\\[0.3cm]
\sg_p(A_1+A_2)^{{\frac{\scalel{1}}{\scalet{p+1}}}}
\le \sg_p(A_1)^{{\frac{\scalel{1}}{\scalet{p+1}}}}
+ \SfG_p(A_2)^{{\frac{\scalel{1}}{\scalet{p+1}}}}.
\end{cases}
\end{align}
It follows from \eqref{eq:trianglep} that 
the functionals $\SfG_p$ and $\sg_p$ are continuous on $\BS_{p, \infty}$:
\begin{align*}
\big|
\SfG_p(A_1)^{{\frac{\scalel{1}}{\scalet{p+1}}}} - \SfG_p(A_2)^{{\frac{\scalel{1}}{\scalet{p+1}}}}
\big|\le &\ \SfG_p(A_1 - A_2)^{{\frac{\scalel{1}}{\scalet{p+1}}}},\\
\big|\sg_p(A_1)^{{\frac{\scalel{1}}{\scalet{p+1}}}} - \sg_p(A_2)^{{\frac{\scalel{1}}{\scalet{p+1}}}}
\big|\le &\ \SfG_p(A_1 - A_2)^{{\frac{\scalel{1}}{\scalet{p+1}}}}.
\end{align*} 
We need the following two corollaries of this fact:

\begin{cor}\label{cor:zero}
Suppose that $\SfG_p(A_1-A_2) = 0$. Then 
\begin{align*}
\SfG_p(A_1) = \SfG_p(A_2),\quad \sg_p(A_1) = \sg_p(A_2).
\end{align*}
\end{cor}

The next corollary is more general:

\begin{cor}\label{cor:zero1}
Suppose that $A\in\BS_{p, \infty}$ and that for every $\nu>0$ there exists an operator 
$A_\nu\in \BS_{p, \infty}$ such that $\SfG_p(A - A_\nu)\to 0$, 
$\nu\to 0$. Then the functionals 
$\SfG_p(A_\nu), \sg_p(A_\nu)$ have limits as $\nu\to 0$ and 
\begin{align*}
\lim_{\nu\to 0} \SfG_p(A_\nu) = \SfG_p(A),\quad 
\lim_{\nu\to 0} \sg_p(A_\nu) = \sg_p(A).
\end{align*}
\end{cor}
 
Remark that using the functionals \eqref{eq:limsupinf} one can rewrite the asymptotics 
\eqref{eq:mainas} in the form
\begin{align}\label{eq:mainas1}
\SfG_{1/2}(\boldsymbol\ST) = \sg_{1/2}(\boldsymbol\ST) = B.
\end{align}

\section{Singular values of integral operators}\label{sect:intop}

In this section we collect some basic results on singular values of integral operators 
that will be instrumental in the proof of the main Theorem 2.4. 
The first such result is borrowed from \cite{BS1977} and it describes    
the membership of integral operators 
in classes $\BS_{p, \infty}$ with some $p>0$ in terms of the smoothness of 
their integral kernels. 
More specifically, we need the bounds in  
$\BS_{p, \infty}$ under the assumption that the integral kernel 
lies in a certain \emph{Besov space}.   
Thus we start with  
basic facts about Besov 
spaces. More details can be found 
in \cite[Chapter 4]{Nikol1975},\cite[Sect. 5.3]{Burenkov1998}, 
\cite[Sect. V.5]{Stein1970} and \cite[Ch. 4]{Triebel1978}. 
 
\subsection{Besov spaces} \label{subsect:nikol} 
For a function $u = u(x)$, $x\in \R^d$, and arbitrary $l = 0, 1, 2, \dots,$ 
define the finite difference  
\begin{align}\label{eq:findif}
\Delta_h^{(l)} u (x) = \sum_{j=0}^l (-1)^{j+l} {l \choose j} u(x+jh),
\end{align}
and the \emph{$\plainL{q}$-modulus of smoothness of order} $l$:
\begin{align}\label{eq:ms}
\om_q^{(l)}(u; t) = \sup_{|h|\le t} \|\Delta_h^{(l)}u\|_{\plainL{q}},\quad t >0.
\end{align}
The function $u$ is said to belong to the Besov space $\plainB^{s}_{q, \infty}(\R^d)$,  
$s>0$, $q\in [1, \infty]$, if 
for some $l > s$ we have 
\begin{align}\label{eq:besov}
\|u\|_{\plainB^{s}_{q, \infty}} := \|u\|_{\plainL{q}} + \sup_{t>0} t^{-s}\om_q^{(l)}(u; t) < \infty.
\end{align}
This quantity defines a norm on 
$\plainB^{s}_{q, \infty}(\R^d)$. 
Different values of $l > s$ lead to equivalent norms \eqref{eq:besov}, 
see \cite[Sect. 5.3]{Burenkov1998}. 
The space $\plainB^{s}_{q, \infty}(\R^d)$ is just one representative of the whole scale 
of Besov spaces $\plainB^{s}_{q, r}(\R^d)$, $0 <r \le \infty$. In this paper we need only 
$r = \infty$. 
Following \cite[Chapter 4]{Nikol1975} we use for the space $\plainB^{s}_{q, \infty}$ 
the notation  $\plainH{s}_{q}$ and write 
$\|u\|_{\plainH{s}_q}$ for the norm \eqref{eq:besov}. Sometimes in the literature this space is called 
the \emph{Nikol'skii space}, see e.g. \cite{BS1977}. 

Note that 
\begin{align}\label{eq:iter}
\Delta_h^{(l)} u(x) = \Delta_h^{(l-m)} (\Delta_h^{(m)} u)(x)
= \sum_{j=0}^{l-m}(-1)^{j+l-m}{l-m \choose j} \Delta^{(m)}_hu(x+jh),
\end{align}
for all $m < l$.
Thus, under the assumption 
$u\in \plainW{l, 1}_{\textup{loc}}(\R^d)$, 
iterating the identity
\begin{align*}
\Delta_h^{(1)} u(x) = u(x+h)- u(x) = \int_0^1 h\cdot\nabla u(x+sh) ds,
\end{align*}
we obtain the formula
\begin{align}\label{eq:fintodiff}
\Delta_h^{(l)} u(x) = \int_0^1 \int_0^1 \cdots \int_{0}^1 
(h\cdot \nabla)^l u\big( x+ \sum_{j=1}^l s_j h\big)\, ds_1 ds_2\dots ds_l.
\end{align}
For a domain $\Om\subset\R^d$ the space $\plainH{s}_q(\Om)$ is defined as the restriction 
of $\plainH{s}_q(\R^d)$ to $\Om$. The corresponding norm of the function 
$u\in \plainH{s}_q(\Om)$ is defined as 
\begin{align}\label{eq:restr}
\|u\|_{\plainH{s}_q(\Om)} = \inf \|g\|_{\plainH{s}_q(\R^d)},
\end{align}
where the infimum is taken over all functions $g\in \plainH{s}_q(\R^d)$ such 
that $u=g$ for a.e. $x\in\Om$. In other words, a function $u$ belongs to 
$\plainH{s}_q(\Om)$ if it has an extension $g\in \plainH{s}_q(\R^d)$. 
Here we have adopted the definitions from \cite[Sect. 4.2.1]{Triebel1978}. 
They will be sufficient for our purposes.  

There are also intrinsic (i.e. not based on restriction) definitions of the Nikol'skii and Besov 
spaces on domains, see for example \cite[Sect. 4.3]{Nikol1975} or \cite[Section 7.32]{Adams_Fournier_2003}. 
For domains 
with ``good" extension properties the intrinsic and restriction definitions are equivalent, see 
\cite[Sect. 7.32]{Adams_Fournier_2003}.  
The collection of such domains includes 
bounded domains satisfying the cone condition (see \cite[Sect. 4.2.3]{Triebel1978}), 
and in particular, bounded Lipschitz domains.

We are interested in functions that fall in the appropriate spaces $\plainH{s}_q$  in a natural way. 
Let $u\in \plainC\infty(\R^d\setminus X)$, where $X = \{a_1, a_2, \dots, a_N\}$ is a finite collection of points $a_k\in \R^d$, $k = 1, 2, \dots, N$. 
Assume that  
for some $\a \in\R$ and some $A\ge 0$ we have 
\begin{align}\label{eq:sing}
|\nabla_x^j u(x)|\lesssim A\big(1+\sum_{k=1}^N \big(|x-a_k|\wedge 1\big)^{\a-j}\big),
\quad \textup{for all} \quad x\notin X\quad \textup{and}
\quad j = 0, 1, \dots,
\end{align} 
where, as explained in the Introduction, $a\wedge b = \min(a, b)$ for any two non-negative $a$ and $b$. 
A good example that ``almost" satisfies this condition 
with $X = \{0\}$ is the class of homogeneous functions, i.e.   
functions 
$\Phi\in\plainC\infty(\R^d\setminus\{0\} )$ such that $\Phi(tx) = t^{\a} \Phi(x)$ for all
$t >0$ and all $x\in \R^d, x\not = 0$, with some $ \a\in\R$. As follows from 
\cite[Lemma 10.1]{BS1977}, if $\a > -d/q$, then $\Phi\in \plainH{s}_q(B)$, $s = \a + d/q$, for 
any ball $B\in\R^d$. 
In the next lemma we show the same for functions satisfying \eqref{eq:sing}.  
This result is elementary but we have not been able to find it in the literature.  
%
%
  
\begin{lem}  \label{lem:nik} 
Suppose that $d\ge 2$ and that the function $u$ satisfies \eqref{eq:sing}.  
If $\a > -d/q$ with some $1\le q\le d$, then 
for all $x_0\in\R^d$ and $R>0$ the function $u$ belongs to $\plainH{s}_q(B)$,  $B = B(x_0, R)$, 
with $s = \a + d/q$.
Moreover, $\|u\|_{\plainH{s}_q(B)}\lesssim A$ where the implicit constant 
does not depend on $x_0$ and the set $X$, but may depend on the radius $R$.  
\end{lem}

We precede the proof with the following useful remark.

\begin{rem}\label{rem:remove}
Define the integer $m = m(\a)$: 
\begin{align}\label{eq:mit}
\begin{cases}
m = 0, \quad \textup{if} \quad \a\le 0,\\[0.2cm]
m = [\a]+1, \quad \textup{if}\quad \a>0, \a\notin \mathbb Z,\\[0.2cm]
m = \a, \quad \textup{if}\quad \a>0, \a\in\Z.
\end{cases}
\end{align}
Under the assumption $d\ge 2$ the potential singularities of the function $u$ 
at the points $a_k, k= 1, 2, \dots, N$, 
are ``removable" in the following sense. 
If $\a >0$ (i.e. $m\ge 1$) the function $u$ clearly belongs to 
$\plainW{m, d}_{\textup{loc}}(\R^d\setminus X)$. Thus it follows from 
Lemma \ref{lem:remove} that for $d\ge 2$ the function $u$ also belongs to 
$\plainW{m, d}_{\textup{loc}}(\R^d)$. 
\end{rem}

\begin{proof}[Proof of Lemma \ref{lem:nik}] 
In view of the definition \eqref{eq:restr} we may assume that 
$u$ is supported on the ball $B(x_0, 2R)$ and shall prove that 
$\|u\|_{\plainH{s}_q(\R^d)}\lesssim A$. 
Since $\a > -d/q$, we have $u\in\plainL{q}(\R^d)$ and $\|u\|_{\plainL{q}}\lesssim A$. 
It remains to 
estimate the modulus of smoothness  \eqref{eq:ms}.

Let 
\begin{align}\label{eq:el}
s = \a + d/q\quad \textup{and} \quad l = [s]+1.
\end{align} 
We shall show that 
\begin{align}\label{eq:ms2}
\|\Delta_h^{(l)} u \|_{\plainL2}\lesssim A\,
(1\wedge|h|^s),\quad s = \a+\frac{d}{2},
\end{align} 
for all $h\in\R^d$, 
with an implicit constant independent of $x_0$ and of the set $X$. 
If $|h|\ge 1$, then by the definition \eqref{eq:findif}, 
\begin{align*}
\|\Delta_h^{(l)} u \|_{\plainL{q}}\le 2^l\|u\|_{\plainL{q}}\lesssim A,
\end{align*}
so that it remains to prove \eqref{eq:ms2} for $|h|\le 1$.

Denote by $\chi_k(x; r), r>0,$ the indicator of the ball $B(a_k, r)$, and define 
$\kappa_k(x; r)= 1-\chi_k(x; r)$. 
The first step is to prove that 
\begin{align}\label{eq:inside}
\|\chi_k(\ \cdot\ ; 2l|h|) \Delta_h^{(l)} u \|_{\plainL{q}}\lesssim A\, |h|^s,\quad 
k = 1, 2, \dots, N.
\end{align}
For the finite difference $\Delta_h^{(l)} u$ we use the representation \eqref{eq:iter} 
where $m$ is as defined in \eqref{eq:mit}. 
%
%
%
In view \eqref{eq:iter}, we have 
\begin{align*}
|\chi_k(x ; 2l|h|)\, \Delta_h^{(l)} u(x) |
\lesssim \sum_{j=0}^{l-m}|\chi_k(x ; 2l|h|)\,  \Delta_h^{(m)} u(x+jh)|.
\end{align*}
Estimate $\plainL{q}$-norms of each term on the right-hand side individually. 
For $m = 0$ we have 
\begin{align}\label{eq:deltam0}
\|\chi_k(\, \cdot\, ; 2l|h|)\,  \Delta_h^{(0)} & u(\ \cdot\ + jh)\|_{\plainL{q}}^q\notag\\
\le &\ 
\int_{\R^d}
\chi_k(x ; 2l|h|)\, \big|u\big(x+ jh\big)\big|^q dx\notag\\
\le &\  \int_{\R^d} \chi_k(x ; 4l|h|)\,
\big|u(x)\big|^q dx\,,  
\end{align}
for all $j = 0, 1, \dots, l$. 

If $m\ge 1$, then in view of Remark 
\ref{rem:remove} 
we have $u\in \plainW{m, d}_{\textup{loc}}(\R^d)$. Therefore, using 
\eqref{eq:fintodiff}, for arbitrary $q \in [1, d]$ we can estimate: 
\begin{align}\label{eq:deltam}
\|\chi_k(\, \cdot\, ; & 2l|h|)\,  \Delta_h^{(m)}   u(\ \cdot\ + jh)\|_{\plainL{q}}^q\notag\\
\le &\ |h|^{qm}  \int_0^1 \int_0^1 \cdots \int_{0}^1\int_{\R^d}
\chi_k(x ; 2l|h|)\, \big|\nabla^m 
u\big(x+ jh + \sum_{n=1}^m s_n h\big)\big|^q dx\, ds_1 ds_2 \dots ds_n\notag\\
\le &\ |h|^{qm}  \int_{\R^d} \chi_k(x ; 4l|h|)\,
\big|\nabla^m u(x)\big|^q dx\,,  
\end{align}
for all $j = 0, 1, \dots, l-m$, where we have used the notation \eqref{eq:nabla}. 
Now use \eqref{eq:sing}, 
so that the right-hand side 
of \eqref{eq:deltam0} (for $m=0$) and \eqref{eq:deltam} (for $m\ge 1$) is bounded by 
\begin{align}\label{eq:polar}
|h|^{qm}\,A^q\,\sum_{n=1}^N\int\limits_{|x-a_k|<4l|h|} \big(1+&\ |x-a_n|^{q(\a-m)}\big)\,  dx
\notag\\[0.2cm]
\lesssim &\ |h|^{qm+d}\,A^q + |h|^{qm}\,A^q\,\sum_{n=1}^N\int\limits_{|x-a_k|<4l|h|} |x-a_n|^{q(\a-m)}\,dx.
\end{align}
%
%
Since $\a-m>-1$ and $d\ge q$, each integral in the above sum is finite.  
To estimate each of the integrals consider two cases. If $a_n\in B(a_k, 5l|h|)$, then 
\begin{align*}
\int\limits_{|x-a_k|<4l|h|} |x-a_n|^{q(\a-m)}\,dx 
\lesssim \int\limits_{|x-a_n|<10l|h|} |x-a_n|^{q(\a-m)}  dx 
\lesssim |h|^{q\a-qm+d}. 
\end{align*}
If $a_n\notin B(a_k, 5l|h|)$, then $|x - a_n|\ge l|h|$ for $x\in B(a_k, 4l|h|)$, 
and as $\a-m \le 0$, we have again
\begin{align*}
\int\limits_{|x-a_k|<4l|h|} |x-a_n|^{q(\a-m)}\, dx 
\lesssim 
%
%
(l|h|)^{q(\a-m)}\,
\int\limits_{|x-a_k|<4l|h|} dx
\lesssim |h|^{q\a-qm+d}. 
\end{align*}
Substituting the last two bounds in \eqref{eq:polar} we conclude that 
%
%
the right-hand side of \eqref{eq:deltam} is bounded from above 
by $A^q\,|h|^{q\a+d}= A^q\,|h|^{qs}$, which leads to 
\eqref{eq:inside}. 

Now, denote 
\begin{align*}
\kappa(x; r) = \prod_k \kappa_k(x; r),\, \kappa_k(x; r) = 1 - \chi_k(x; r),
\end{align*}
and prove that  
\begin{align}\label{eq:outside}
\|\kappa(\ \cdot\ ; 2l|h|)\Delta_h^{(l)}u\|_{\plainL{q}}\lesssim A\,|h|^{s},
\end{align}
where the number $l$ is defined in \eqref{eq:el}. 
The function $u$ is $\plainC\infty$ on the support of the function $\kappa$. Therefore, 
similarly to \eqref{eq:deltam}, 
\begin{align*}
\|\kappa(\ \cdot\ ; 2l|h|)\, \Delta_h^{(l)} & u\|_{\plainL{q}}^q\\
\le &\ |h|^{ql}  \int_0^1 \int_0^1 \cdots \int_{0}^1\int_{\R^d}
\kappa(x; 2l|h|)\,\big|\nabla^l 
u\big(x+ \sum_{n=1}^l s_n h\big)\big|^q dx\, ds_1 ds_2 \dots ds_n\\
\le &\ |h|^{ql}  \int 
\kappa(x; l|h|)
\, \big|\nabla^l u(x)\big|^q dx\,. 
\end{align*}
Since $ql > q\a+d$, by \eqref{eq:fintodiff}, the right-hand side does not exceed 
\begin{align*}
|h|^{q l}\, A^q\,\sum_{k=1}^N  \int\limits_{\substack{|x-a_k|> l|h|\\ x\in B(x_0, 2R)}} 
\big(1+|x-a_k|^{q(\a-l)}\big) \, dx\lesssim \, 
|h|^{ql}A^q\big(1 + |h|^{q\a-ql + d}\big)
\lesssim 
A^q\,|h|^{qs}.
\end{align*}
This implies \eqref{eq:outside}. 

Using the inequality $1\le \sum_{k=1}^N \chi_k(x; r) + \kappa(x; r)$, we can estimate
\begin{align*}
\|\Delta_h^{(l)}u\|_{\plainL{q}}^q
\le \sum_{k=1}^N\|\chi_k(\ \cdot\ ; 2l|h|)\, \Delta_h^{(l)}u\|_{\plainL{q}}^q 
+ \|\kappa(\ \cdot\ ; 2l|h|)\Delta_h^{(l)}u\|_{\plainL{q}}^q,
\end{align*}
and hence putting \eqref{eq:inside} and \eqref{eq:outside} together we obtain \eqref{eq:ms2}. 
Along with the fact that $\|u\|_{\plainL{q}}\lesssim A$ this implies that 
$\|u\|_{\plainH{s}_q}\lesssim A$ as required. 
\end{proof}

\subsection{Bounds for singular values of integral operators}\label{subsect:BS}
Let $X\subset \R^d$, be a bounded Lipschitz domain 
and let $a = a(x),\, x\in X$ and $b = b(t),\, t\in \R^l$, be some functions that we call 
weights. We consider the integral operator 
$T_{ba}:\plainL2(X)\to\plainL2(\R^l)$ of the form
\begin{align*}
(T_{ba}u)(t) = b(t) \int_X T(t, x) a(x) u(x)\,dx.
\end{align*}
For estimates of the 
singular values we rely on \cite[Corollaries 4.1 and 4.5, Theorems 4.4, 4.7]{BS1977} 
which we state here in the form convenient for 
our purposes. 

\begin{prop}\label{prop:BS} 
Let $X\subset \R^d$ be a bounded Lipschitz domain. 
Assume that the kernel $T(t,x), t\in \R^l,\, x\in X,$ is such that 
$T(t, \, \cdot\, )\in \plainH{s}_2(X)$ with some $s>0$, 
for a.e. $t\in \R^l$. Assume that 
$b\in\plainL2_{\textup{\tiny loc}}(\R^l)$ and that 
$a\in\plainL{r}(X)$, where 
\begin{align*}
\begin{cases}
r = 2, \quad \textup{if} \quad 2s >d,\\
r>2, \quad \textup{if}\quad 2s = d,\\
r > d s^{-1}, \quad \textup{if}\quad 2s <d. 
\end{cases}
\end{align*}
Then $T_{ba}\in\BS_{q, \infty}$ where $q^{-1} = 2^{-1} + s d^{-1}$ and 
\begin{align*}
\|T_{ba}\|_{q, \infty}\lesssim \biggl[\int_{\R^l}  
\|T(t,\ \cdot\ )\|_{\plainH{s}_2(X)}^2\, |b(t)|^2 \,dt \biggr]^{\frac{1}{2}}
\|a\|_{\plainL{r}(X)},
\end{align*}
under the assumption that the right-hand side is finite.
\end{prop}    

We apply this proposition to the kernel $T(t, x)$ that satisfies the 
bound \eqref{eq:sing} 
with some $\a>-d/2$ as a function of $x$. 
In this case we define $s = \a + d/2>0$, so that the conditions on the 
number $r$ in Proposition 
\ref{prop:BS} rewrite as follows  
\begin{align}\label{eq:r}
\begin{cases}
r = 2, \quad \textup{if} \quad \a >0,\\[0.2cm]
r>2, \quad \textup{if}\quad \a = 0,\\[0.2cm]
r > 2d(2\a+d)^{-1}, \quad \textup{if}\quad -d < 2\a <0. 
\end{cases}
\end{align}
Below we denote $\CC_n = [0, 1)^d + n, n\in\mathbb Z^d$. The notation $\1_n$ is used for the 
indicator function of the cube $\CC_n$.  

In the next theorem we consider $T_{ba}$ as an operator from $\plainL2(\R^d)$ to $\plainL2(\R^l)$. 

\begin{thm} \label{thm:BSspace} 
Let $d\ge 2$ and $l\ge 1$. 
Let $z_k: \R^l\mapsto \R^d$,\, 
$k = 1, 2, \dots, N$, be some functions of $t\in\R^l$, and let $\a > -d/2$. 
Consider the kernel $T(t, x)$, $x\in\R^d, t\in\R^l,$ satisfying the condition 
\begin{align}\label{eq:kernelexp}
|\nabla_x^j T(t, x)| \le A(t, x) 
\big(1+\sum_{k=1}^N \big(1\wedge |x-z_k(t)|\big)^{\a-j}\big),\quad  
j = 0, 1, \dots, 
\end{align}    
for a.e. $t\in\R^l$ and a.e. $x\in \R^d$, where 
$A\in\plainL\infty_{\textup{\tiny loc} }(\R^l\times\R^d)$ 
is a non-negative function. 
%
%
Let $b\in\plainL2_{\textup{loc}}(\R^l)$ 
and $a\in\plainL{r}_{\textup{\tiny loc}}(\R^d)$ 
where the parameter $r$ is as specified in \eqref{eq:r}.
Then the operator $T_{ba}$ 
belongs to $\BS_{q, \infty}$ where $q^{-1} = 1+\a d^{-1}$ and 
\begin{align}\label{eq:piv}
\|T_{ba}\|_{q, \infty}\lesssim 
\bigg[\sum_{n\in\Z^d}
\|A_n b\|_{\plainL2(\R^l)}^q\,\|a\|_{\plainL{r}(\CC_n)}^q
\bigg]^{\frac{1}{q}},\quad 
A_n(t) = \|A(t,\,\cdot\,)\|_{\plainL\infty(\CC_n)},
\end{align}
under the assumption that the right-hand side of the above inequality is finite. 
\end{thm}

\begin{proof} 
Consider the operator $W_n = T_{ba}\1_n$, $n\in\mathbb Z^d$,     
with kernel $W_n(t, x) = b(t)T(t, x) a(x)\1_n(x)$ 
and apply Proposition \ref{prop:BS} to each $W_n$. 
Since $d\ge 2$, by \eqref{eq:kernelexp} it follows from 
Lemma \ref{lem:nik} that 
$W_n(t,\, \cdot\, )\in \plainH{s}_2(\CC_n)$, $s = \a+d/2$,  
for a.e. $t\in\R^l$, and 
\begin{align*}
\|W_n(t, \, \cdot\, )\|_{\plainH{s}_2(\CC_n)}\lesssim A_n(t). 
\end{align*}
Thus, by Proposition \ref{prop:BS}, the operator $W_n$ belongs to 
$\BS_{q, \infty}$ with $1/q = 1/2 + s/d = 1+ \a/d$, and 
\begin{align}\label{eq:wn}
\| W_n\|_{q, \infty}\lesssim \bigg[\int_{\R^l} A_n(t)^2 |b(t)|^2 dt \bigg]^{\frac{1}{2}}
 \|a\|_{\plainL{r}(\CC_n)} = \|A_n b\|_{\plainL2(\R^m)}\|a\|_{\plainL{r}(\CC_n)},
\end{align} 
with a constant independent of $n$. 
Now observe that $W_k W_j^* = 0$ for $k\not = j$. Furthermore, 
since $\a > -d/2$ we have $q < 2$, and hence by Proposition \ref{prop:orthog},
\begin{align*}
\| T_{ba}\|_{q, \infty}^q\le 2(2-q)^{-1} \sum_{n\in\mathbb Z^n} \|W_n\|_{q, \infty}^q.
\end{align*}
Using \eqref{eq:wn} we obtain the required bound \eqref{eq:piv}.
\end{proof}
 
\subsection{Model operator}\label{subsect:model} 
When studying the spectral asymptotics of the operator $\boldsymbol\ST$, we make use of the following model operator that was examined  in \cite{Sobolev2021}. 

Let $a, b_{j,k}, \b_{j,k}$, $j = 1, 2, \dots, N$, $k = 1, 2, \dots, N-1$, be 
scalar functions such that 
\begin{align}\label{eq:abbeta} 
\begin{cases}
a\in\plainC\infty_0(\R^3), &\ \quad b_{j,k}\in \plainC\infty_0(\R^{3N-3}),\\[0.2cm] 
\b_{j,k}\in\plainC\infty(\R^{3N}), 
\end{cases}
\end{align}
for all $j = 1, 2, \dots, N$, $k = 1, 2, \dots, N-1$. 
Let $\Phi\in\plainC\infty(\R^3\setminus \{0\})$ be a vector 
function with $m$ scalar components, 
homogeneous of order $\a>-3$:
\begin{align}\label{eq:phi}
\Phi(t x) = t^\a \Phi(x),\quad x\not = 0, t >0.
\end{align} 
Consider the vector-valued kernel $\CM(\hat\bx, x)$: 
\begin{align}\label{eq:cm}
\begin{cases}
\CM(\hat\bx, x) = \{\CM_j(\hat\bx, x)\}_{j=1}^N,\ 
\quad
\CM_j(\hat\bx, x) = \sum_{k=1}^{N-1} \CM_{j,k}(\hat\bx, x),\\[0.3cm]
\CM_{j,k}(\hat\bx, x) =  b_{j,k}(\hat\bx) \Phi(x-x_k) a(x)\b_{j,k}(\hat\bx, x). 
\end{cases}
\end{align}
The eigenvalue asymptotics for the operator 
$\iop(\CM):\plainL2(\R^3)\mapsto \plainL2(\R^{3N-3}; \mathbb C^{mN})$ 
was found in \cite{Sobolev2021} for 
general functions $\Phi$ of the form \eqref{eq:phi}. Below we state this result for the 
function $\Phi(x) = \nabla |x| = x|x|^{-1}$, which is 
homogeneous of order $0$. This is 
the case needed for the study of the operator 
$\boldsymbol\ST$.  
We use the representations $(\hat\bx, x) = (\tilde\bx_{k, N}, x_k, x)$ introduced in \eqref{eq:xtilde}.  
Denote  
\begin{align}\label{eq:hm}
\begin{cases}
h(t) = \bigg[\sum_{j=1}^N\sum_{k=1}^{N-1}\int_{\R^{3N-6}} 
| b_{j,k}(\tilde\bx_{k, N}, t)\b_{j,k}(\tilde\bx_{k, N}, t, t) |^2 d\tilde\bx_k\bigg]^{\frac{1}{2}},\ 
\textup{if}\ N\ge 3;\\[0.3cm]
h(t) = \big(|b_{1,1}(t)\b_{1,1}(t, t)|^2 
+ |b_{2,1}(t)\b_{2,1}(t, t)|^2\big)^{\frac{1}{2}},\ \textup{if}\ N=2.
\end{cases}
\end{align}
The next proposition summarizes the required results from \cite{Sobolev2021}. 

\begin{prop}\label{prop:gradp} 
If $\Phi(x)=\nabla|x| = x|x|^{-1}$, then $\iop(\CM)\in \BS_{1, \infty}$ and 
 \begin{align}\label{eq:gradp} 
\SfG_1\big(\iop(\CM)\big)
= \sg_1\big(\iop(\CM)\big)
= \nu_{0, 3}\int_{\R^3} |a(x) h(x)|\, dx,
\end{align}
where $\nu_{0, 3} = 4(3\pi)^{-1}$.

If $\Phi(x)$ is homogeneous of order $\a >0$, then $\SfG_1\big(\iop(\CM)\big) = 0$.
\end{prop}

\section{Factorization of $\boldsymbol \ST$: 
operator $\boldsymbol\SV$}\label{sect:factor}

 \subsection{Factorization of $\boldsymbol\ST$} 
\label{subsect:fact}
Let us describe the factorization of the operator $\boldsymbol\ST$ 
which is central to the proof of Theorem 
\ref{thm:maincompl}. 
Introduce the new set of functions
\begin{align*}
\psi_j(\hat\bx, x) = 
\psi(x_1, \dots, x_{j-1}, x, x_j, \dots, x_{N-1}),\quad j = 1, 2, \dots, N,
\end{align*}
depending on the fixed pair variables $\hat\bx$ and $x$. 
Thus the definition \eqref{eq:tau} rewrites in the form:
\begin{align}\label{eq:psij}
\tau(x, y) = &\ \sum_{j=1}^N \int_{\R^{3N-3}} 
\overline{\sv_j(\hat\bx, x)} \cdot \sv_j(\hat\bx, y) d\hat\bx,
\quad \textup{where} 
\notag\\
\sv_j(\hat\bx, x) = &\ \nabla_x\psi_j(\hat\bx, x),\quad j = 1, 2, \dots, N.
\end{align}
Therefore $\boldsymbol\ST$ can be represented as a product $\boldsymbol\ST = \boldsymbol\SV^*\boldsymbol\SV$, 
where $\boldsymbol\SV:\plainL2(\R^3)\mapsto\plainL2(\R^{3N-3};\mathbb C^{3N})$ 
is the integral operator with the vector-valued kernel $\SV(\hat\bx, x)$ given by  
\begin{align}\label{eq:bpsi}
\SV(\hat\bx, x) = \{\sv_j(\hat\bx, x)\}_{j=1}^N. 
\end{align} 
Given such a factorization, by formula \eqref{eq:double} 
the asymptotic relation \eqref{eq:mainas1} translates to the formula 
\begin{align}\label{eq:ttov}
\SfG_1(\boldsymbol\SV) = \sg_1(\boldsymbol\SV) = B.
\end{align}
Thus the main Theorem \ref{thm:maincompl} can be recast as follows:

\begin{thm}\label{thm:ttov}
Under the conditions of Theorem \ref{thm:maincompl} the 
formula \eqref{eq:ttov} holds with the constant $B$ 
which is defined in \eqref{eq:coeffA}.
\end{thm}

Before proceeding with the proof, we have to rephrase the regularity conditions given 
in Sect. \ref{sect:reg} in terms of the functions $\psi_j$. We start with the bound 
\eqref{eq:FS} which takes the form 
\begin{align}\label{eq:FSj}
|\nabla_{x}^m\psi_j(\hat\bx, x)|\lesssim 
\big(1+\dc(\hat\bx, x)^{1-m}\big)e^{-\varkappa {|\bx|_\scale}}, 
\, \quad \textup{for all}\quad j = 1, 2, \dots, N,\, m = 0, 1, \dots,\, 
\end{align}
and all $(\hat\bx, x)\in\SPi_N$, where 
\begin{align}\label{eq:dist}
\dc(\hat\bx, x) = \min \{1, |x_k-x|, \ 0\le k \le N-1\}.
\end{align}
(Recall the notation $x_0 = 0$.)
In order to rephrase Proposition \ref{prop:repr} a few more definitions are required. 
Let $j\not = k$, and let $\Om_{j, k}$ be the sets and 
$\xi_{j, k}(\bx)$, $\eta_{j, k}(\bx)$ be the functions
from Proposition \ref{prop:repr}.
For all $j = 1, 2, \dots, N$ and all $k = 0, 1, \dots, N-1,$ denote
\begin{align*}
\tilde\Om_{j, k} = 
\begin{cases}
\{(\hat\bx, x)\in\R^{3N}: 
(x_1, \dots, x_{j-1}, x, x_{j}, \dots, x_{N-1})\in\Om_{j, k}\},\quad \textup{if}\ j\ge k+1,\\
\{(\hat\bx, x)\in\R^{3N}: 
(x_1, \dots, x_{j-1}, x, x_{j}, \dots, x_{N-1})\in\Om_{j, k+1}\},\quad \textup{if}\ j\le k.
\end{cases}
\end{align*} 
According to \eqref{eq:omin} we have  
\begin{align}\label{eq:tom}
\SU_{N,k}^{(\dc)}\subset \tilde\Om_{j, k} \subset \SU_{N,k},\quad k = 0, 1, \dots, N-1,
\end{align}  
for all $j = 1, 2, \dots, N$. 
Together with functions $\xi_{j,k}, \eta_{j,k}$ define 
\begin{align*}
\tilde\xi_{j, k}(\hat\bx, x) = 
\begin{cases}
\xi_{j,k}(x_1, \dots, x_{j-1}, x, x_{j}, \dots, x_{N-1}),\quad \textup{if}\ j\ge k+1,\\[0.2cm]
\xi_{j,k+1}(x_1, \dots, x_{j-1}, x, x_{j}, \dots, x_{N-1}),\quad \textup{if}\ j\le k,
\end{cases}
\end{align*}
and 
\begin{align*}
\tilde\eta_{j,k}(\hat\bx, x) = 
\begin{cases}
\eta_{j,k}(x_1, \dots, x_{j-1}, x, x_{j}, \dots, x_{N-1}),\quad \textup{if}\ j\ge k+1,\\[0.2cm]
\eta_{j,k+1}(x_1, \dots, x_{j-1}, x, x_{j}, \dots, x_{N-1}),\quad \textup{if}\ j\le k.
\end{cases}
\end{align*}
By Proposition \ref{prop:repr}, 
for each $j = 1, 2, \dots, N$, and each $k = 0, 1, \dots, N-1$, we have 
\begin{align}\label{eq:psirepr}
\psi_j(\hat\bx, x) = 
\tilde\xi_{j, k}(\hat\bx, x) + |x_k-x| \,\tilde\eta_{j, k}(\hat\bx, x), 
\quad \textup{for all}\ (\hat\bx, x) \in \tilde\Om_{j, k}.
\end{align} 
Observe that the sets $\tilde\Om_{j, k}$ 
and the functions $\tilde\xi_{j,k}, \tilde\eta_{j,k}$  are 
not symmetric under the permutation $j\leftrightarrow k$.

The function \eqref{eq:H} can be easily rewritten via the new functions $\tilde\eta_{j, k}$:
\begin{align}\label{eq:tH}
H(x) =  
\begin{cases}
\big(|\tilde\eta_{1, 1}(x, x)|^2 + |\tilde\eta_{2, 1}(x, x)|^2\big)^{1/2},\quad \textup{if}\ N = 2;\\[0.3cm]
\left[\sum\limits_{j=1}^N\sum\limits_{k=1}^{N-1}\int_{\R^{3N-6}} \big| 
\tilde\eta_{j, k}(\tilde\bx_{k,N}, x, x) \big|^2 d\tilde\bx_{k, N}\right]^{\frac{1}{2}},\quad 
\textup{if}\ N\ge 3.
\end{cases}
\end{align}   
This calculation is done in \cite[Sect. 2]{Sobolev2021}.

The rest of the paper is focused on the proof of Theorem \ref{thm:ttov}. 
As explained in the Introduction, 
at the heart of the proof lies the formula \eqref{eq:repr} for the function $\psi$, 
which translates to the representation \eqref{eq:psirepr}
for the functions $\psi_j$. 
This representation allows us to reduce the problem to the model operator 
considered in Subsect. \ref{subsect:model}.  
At the first stage of this reduction we construct 
convenient approximations of the kernels $\sv_j(\hat\bx, x) = \nabla_x\psi_j(\hat\bx, x)$.

\subsection{Cut-off functions} 
First we construct appropriate cut-offs. 
Fix a $\d > 0$. Along with the sets \eqref{eq:sls} introduce 
\begin{align}\label{eq:slsd}
\SfS_{l, s}(\d) = \SfS_{s, l}(\d) = \{\bx\in\R^{3N}: |x_l-x_s|>\d\},\ 0\le l < s\le N,
\end{align} 
and for all $k = 0, 1, \dots, N-1$, define
\begin{align*}
\SU_{k}(\d) = \bigg(\bigcap_{0\le l < s\le N-1} \SfS_{l, s}(\d)\bigg)\bigcap\bigg(
\bigcap_{\substack{0\le s\le N-1\\s\not = k}} \SfS_{s, N}(\d)\bigg).
\end{align*}
Comparing with \eqref{eq:uj} we see that $\SU_k(\d)\subset \SU_{k, N}$, 
and for $\bx\in \SU_k(\d)$ all 
the coordinate pairs, except for $x_k$ and $x_N$, are separated by a distance $\d$.
Similarly to \eqref{eq:diag} define the diagonal set 
\begin{align*}
\SU^{(\rm d)}_{k}(\d) = \{\bx\in\SU_{k}(\d): x_k = x_N\}\subset\SU_{k, N}^{(\dc)}.
\end{align*}
Recall that the representation 
\eqref{eq:trepr} holds on the domain $\tilde\Om_{j, k}$ which satisfies 
\eqref{eq:tom} for all $j = 1, 2, \dots, N$, $k = 0, 1, \dots, N-1$. 
We construct a compact subset of $\tilde\Om_{j, k}$ in the following way. 
For $R>0$ let 
\begin{align*}
\SU_{k}(\d, R) = &\  \SU_k(\d)\bigcap \, B_R,\\ 
\SU^{(\rm d)}_{k}(\d, R) = &\ \{\bx\in\SU_{k}(\d, R): x_k = x_N\},
\end{align*}
where 
$B_R = \bigtimes_{j=1}^N \{x_j\in\R^3: |x_j| <R\}$. 
The  set $\SU^{(\rm d)}_{k}(\d, R)$ is bounded and its closure belongs to 
$\tilde\Om_{j,k}$ for all $\d>0, R>0$. Therefore, there exists 
an $\varepsilon_0 = \varepsilon_0(\d, R)>0$ such that the $j$-independent $\varepsilon$-neighbourhood 
\begin{align}\label{eq:omj}
\tilde\Om_{k}(\d, R, \varepsilon) := \{\bx\in\SU_{k}(\d, R): |x_k-x_N|<\varepsilon\},
\end{align}
together with its closure, belongs to $\tilde\Om_{j,k}$ 
for all $\varepsilon\in (0, \varepsilon_0)$:
\begin{align}\label{eq:omdere}
\overline{\tilde\Om_{k}(\d, R, \varepsilon)}\subset \tilde\Om_{j,k},\quad \forall 
\varepsilon\in (0, \varepsilon_0).
\end{align} 
Now we specify $\plainC\infty_0$ cutoffs supported on the domains 
$\tilde\Om_{k}(\d, R, \varepsilon)$. 
Let $\t\in\plainC\infty_0(\R)$ and $\z = 1-\t$ be as defined in \eqref{eq:sco}, \eqref{eq:sco1}. 
Denote 
\begin{align}\label{eq:ydel}
Y_\d(\hat\bx) = \prod_{0\le l < s \le N-1} \z\big(|x_l-x_s|(4\d)^{-1}\big).
\end{align}
By the definition of $\z$, 
\begin{align*}
\supp Y_\d\subset 
\bigcap_{0\le l<s\le  N-1} \SfS_{l, s}(2\d),
\end{align*}
where $\SfS_{l, s}(\ \cdot\ )$ is defined in \eqref{eq:slsd}.  
As $\t \ge 0$ and $\z\le 1$, we have 
\begin{align}\label{eq:comply}
1-Y_\d(\hat\bx)\le \sum_{0\le l<s\le N-1} \t\big(|x_l-x_s|(4\d)^{-1}\big),
\end{align}
Furthermore, due to \eqref{eq:sco1}, for any $\varepsilon\le \d$ we have 
\begin{align}\label{eq:partun}
Y_\d(\hat\bx)\sum_{j=0}^{N-1} \t\big(|x-x_j|\varepsilon^{-1}\big) 
+ Y_\d(\hat\bx)\prod_{j=0}^{N-1} \z\big(|x-x_j|\varepsilon^{-1}\big) = Y_\d(\hat\bx).
\end{align} 
Define also cut-offs at infinity. Denote
\begin{align}\label{eq:qr} 
Q_R(\hat\bx) = \prod_{1\le l\le N-1} \t\big(|x_l|R^{-1}\big),\quad 
K_R(x) = \t\big(|x|R^{-1}\big).
\end{align}
As shown in \cite[Lemma 5.2]{Sobolev2021}, 
\begin{align}\label{eq:cutoff}
\textup{for all}\ \ &\ \varepsilon<\min\{\varepsilon_0, \d\}\quad 
\textup{the support of the function}\notag\\
&\ Q_R(\hat\bx) K_R(x) Y_\d(\hat\bx) \t\big(|x-x_k|\varepsilon^{-1}\big)
\ \  
\textup{belongs to}\ \ \tilde\Om_k(\d, R, \varepsilon)
\end{align} 
for all $k = 0, 1, \dots, N-1$. 
 
\subsection{} 
Using the cut-offs introduced above we construct an  
 approximation 
for the kernels $\sv_j(\hat\bx, x)$. It follows from \eqref{eq:psirepr} that 
\begin{align}\label{eq:trepr}
\sv_j(\hat\bx, x) = 
\nabla_x\tilde\xi_{j, k}(\hat\bx, x) + &\ |x_k-x|\, \nabla_x \tilde\eta_{j, k}(\hat\bx, x) \notag\\
&\ + (x-x_k)|x_k-x|^{-1} \tilde\eta_{j, k}(\hat\bx, x),
\quad \textup{for all}\ (\hat\bx, x) \in \tilde\Om_{j, k}.
\end{align} 
The last term in the above sum is the one that determines the asymptotics 
for the operator $\boldsymbol\SV$. To construct for it a suitable approximation we assume that 
the parameter $\varepsilon_0(\d, R)$ in \eqref{eq:omdere} satisfies 
$\varepsilon_0(\d, R) \le \d$, so that   
for all 
$\varepsilon < \varepsilon_0(\d, R)$, 
apart from the inclusion \eqref{eq:omdere} 
we also have \eqref{eq:cutoff}. 
Thus, for these values of $\varepsilon$ 
the real analytic functions $\tilde\xi_{j,k}, \tilde\eta_{j,k}$ 
are well-defined on the support of the function \eqref{eq:cutoff}, 
and hence the kernel 
\begin{align}\label{eq:upsilon}
\Upsilon_j[\d,R,\varepsilon](\hat\bx, x) 
=  Q_R(\hat\bx) Y_\d(\hat\bx) K_R(x)\sum_{k=1}^{N-1}\t\big(|x-x_k|\varepsilon^{-1}\big)
\frac{x-x_k}{|x-x_k|}\tilde\eta_{j, k}(\hat\bx, x),
\end{align}
is well-defined for all $(\hat\bx, x)\in\R^{3N}$, and each of the functions
\begin{align*}
Q_R(\hat\bx) Y_\d(\hat\bx) K_R(x) \t\big(|x-x_k|\varepsilon^{-1}\big)
 \tilde\eta_{j,k}(\hat\bx, x),\quad k = 1, 2, \dots, N-1,
\end{align*}
is $\plainC\infty_0(\R^{3N})$. 
Our objective is to prove that the vector-valued kernel 
\begin{align*}
\boldsymbol\Upsilon[\d, R, \varepsilon](\hat\bx, x) 
= \big\{\Upsilon_j[\d,R,\varepsilon](\hat\bx, x)\big\}_{j=1}^N
\end{align*} 
is indeed 
an approximation for $\boldsymbol\SV(\hat\bx, x)$(see \eqref{eq:bpsi}) in the following sense. 

\begin{lem}\label{lem:central} 
The following relations hold: 
\begin{align*}
\SfG_1(\boldsymbol\SV) = \lim\limits_{\substack{\d\to 0\\ R\to\infty}} \lim_{\varepsilon\to 0}
\SfG_1\big(\iop(\boldsymbol\Upsilon[\d, R, \varepsilon])\big),\quad 
\sg_1(\boldsymbol\SV) = \lim\limits_{\substack{\d\to 0\\ R\to\infty}} \lim_{\varepsilon\to 0}
\sg_1\big(\iop(\boldsymbol\Upsilon[\d, R, \varepsilon])\big), 
\end{align*}
where the limits on the right-hand side exist. 
\end{lem}

Our next step is prove the main theorems (Theorems 
\ref{thm:etadiag}, \ref{thm:ttov} and \ref{thm:maincompl}) 
assuming  Lemma \ref{lem:central}. The proof of Lemma \ref{lem:central} 
is postponed until Sect. \ref{sect:trim}.   

\subsection{Proof of Theorems \ref{thm:etadiag} and \ref{thm:ttov}, \ref{thm:maincompl}}
\label{subsect:proofs}

First we apply Proposition \ref{prop:gradp} to the operator 
$\iop\big(\boldsymbol\Upsilon[\d, R, \varepsilon]\big)$. 

\begin{lem} 
The operator $\iop(\boldsymbol\Upsilon[\d, R, \varepsilon])$ belongs to $\BS_{1, \infty}$ 
for all $\d >0, R>0, \varepsilon<\varepsilon_0(\d, R)$ 
 and 
\begin{align}\label{eq:upsas}
\SfG_1\big(\iop(\boldsymbol\Upsilon[\d, R, \varepsilon])\big) 
= \sg_1\big(\iop(\boldsymbol\Upsilon[\d, R, \varepsilon])\big) 
= \frac{4}{3\pi} \int  K_R(t) H_{\d, R}(t) dt,
\end{align}
where
\begin{align*}
H_{\d, R}(t) = Q_R(t) Y_\d(t)\, \big(|\tilde\eta_{1, 1}(t, t)|^2 + \tilde\eta_{1, 2}(t, t)|^2\big)^{1/2},\ \textup{if}\ N=2,
\end{align*}
and
\begin{align}\label{eq:hdr}
H_{\d, R}(t) =  
\bigg[\sum_{j=1}^N\sum_{k=1}^{N-1}\int_{\R^{3N-6}} \big| 
Q_R(\tilde\bx_{k, N}, t) Y_\d(\tilde\bx_{k, N}, t)
\tilde\eta_{j, k}(\tilde\bx_{k, N}, t, t) \big|^2 d\tilde\bx_{k, N}\bigg]^{\frac{1}{2}},\ \textup{if}\ N\ge 3.
\end{align}  
\end{lem} 
 
\begin{proof}
The kernel $\boldsymbol\Upsilon[\d, R, \varepsilon]$ (see \eqref{eq:upsilon}) has the form \eqref{eq:cm} 
with 
\begin{align*}
a(x) = K_{2R}(x),\ &\ b_{j,k}(\hat\bx) = Q_{2R}(\hat\bx) Y_{\d/2}(\hat\bx),\\ 
\b_{j, k}(\hat\bx, x) = &\ \t\big(|x-x_k|\varepsilon^{-1}\big) \tilde\eta_{j, k}(\hat\bx, x)
Q_{R}(\hat\bx) Y_{\d}(\hat\bx)K_{R}(x),
\end{align*}
and the homogeneous function $\Phi(x) = x|x|^{-1}$.   
Here we have used the fact that 
\begin{align*}
Q_R(\hat\bx)Q_{2R}(\hat\bx) = Q_R(\hat\bx),\quad 
Y_\d(\hat\bx) Y_{\d/2}(\hat\bx) = Y_\d(\hat\bx)
\quad \textup{and} \quad K_R(x) K_{2R}(x) = K_R(x).
\end{align*}
Thus defined functions $a, b_{j, k}$ and $\b_{j, k}$ satisfy conditions \eqref{eq:abbeta}. 
Therefore we can use Proposition \ref{prop:gradp}. It is immediate to see that in this case 
the function $h$ defined in \eqref{eq:hm}, coincides with $H_{\d, R}$, so that 
\eqref{eq:gradp} entails \eqref{eq:upsas}, as required. 
\end{proof}

\begin{proof}[Proof of Theorems \ref{thm:etadiag}, \ref{thm:ttov} and \ref{thm:maincompl}]
By Lemma \ref{lem:central},  
the r.h.s of  
the relation \eqref{eq:upsas} has a finite 
limit as $\d\to 0$ and $R\to\infty$, and  
\begin{align}\label{eq:pre}
\SfG_1(\boldsymbol\SV) = \sg_1(\boldsymbol\SV) = \lim_{\substack{\d\to 0\\ R\to\infty}} \,
\frac{4}{3\pi} \int  K_R(t) H_{\d, R}(t) dt.
\end{align} 
Therefore the integral on the right-hand side of 
\eqref{eq:pre} is bounded uniformly in $\d>0$ and $R>0$. 
Assume for convenience that the function $\t$ defined 
in \eqref{eq:sco1} is monotone decreasing on the half-axis 
$[0, \infty)$. Therefore the pointwise convergences
\begin{align*}
Y_\d(\tilde\bx_{k, N}, t)\to 1, \ \d\to 0\quad \textup{and}\quad 
K_R(t)\to 1, Q_R(\tilde\bx_{k, N} , t)\to 1,\ R\to\infty,
\end{align*} 
are monotone increasing. By the Monotone Convergence Theorem, 
the integrand $ K_R(t) H_{\d, R}(t) $ on the right-hand side of \eqref{eq:pre} 
converges for a.e. $t\in\R^3$ as $\d\to 0, R\to\infty,$ to an $\plainL1(\R)$-function,  
which we denote by $\tilde H(t)$, and the 
integral in \eqref{eq:pre} converges to 
\begin{align*}
\frac{4}{3\pi}\int \tilde H(t) dt.  
\end{align*}
If $N=2$, then this concludes the proof 
of Theorem 
\ref{thm:etadiag}, 
since in this case 
\begin{align*}
H_{\d, R}(t)\to\big(|\tilde\eta_{1, 1}(t, t)|^2 + \tilde\eta_{2, 1}(t, t)|^2\big)^{1/2},
\end{align*} 
a.e. $t\in\R^3$, and by virtue of \eqref{eq:tH} this limit coincides with $H(t)$. 

If $N\ge 3$, then the convergence to $\tilde H(t)$ implies that  
for a.e. $t\in \R^3$ 
the function $K_R(t)H_{\d, R}(t)$, and hence $H_{\d, R}(t)$, is bounded uniformly in $\d$ and $R$. 
Applying the Monotone Convergence Theorem to the integral \eqref{eq:hdr}, 
we conclude that 
for a.e. $t\in\R^3$ 
the a.e.-limit
\begin{align*}
|\tilde\eta_{j,k}(\tilde\bx_{k, N}, t, t)| 
= \lim_{\substack{\d\to 0\\ R\to\infty}}\big| Q_R(\tilde\bx_{k, N}, t) Y_\d(\tilde\bx_{k, N}, t)
\tilde\eta_{j,k}(\tilde\bx_{k, N}, t, t) \big|,\ 
\end{align*}
belongs to $\plainL2(\R^{3N-6})$, and 
\begin{align*}
\lim_{\substack{\d\to 0\\ R\to\infty}} \,H_{\d, R}(t) = H(t),\quad \textup{a.e.}\quad t\in \R^3,
\end{align*}
 where we have used the formula \eqref{eq:tH} for $H$. Thus 
 $H = \tilde H\in\plainL1(\R^3)$. As 
\eqref{eq:tH} is equivalent to \eqref{eq:H}, 
this completes the proof of Theorem \ref{thm:etadiag}. 
 Furthermore, since the limit on the right-hand side of 
 \eqref{eq:pre} coincides with the coefficient 
 $B$ defined in \eqref{eq:coeffA}, this also completes the proof of Theorem 
 \ref{thm:ttov}.
 As explained before,  Theorem \ref{thm:ttov} is 
 equivalent to Theorem \ref{thm:maincompl}. 
\end{proof}  

The rest of the paper focuses on the proof of Lemma \ref{lem:central}. 

\section{Proof of Lemma \ref{lem:central}}\label{sect:trim}   

The pivotal role in the proof 
is played by the singular value estimates for the weighted operator 
$\boldsymbol\SV$, that are derived in the next subsection.

\subsection{Spectral bound for the operator $\boldsymbol\SV$ }\label{subsect:prep}
Recall that the vector-valued kernel $\SV(\hat\bx, x)$ of the operator 
$\boldsymbol\SV$ is given by \eqref{eq:bpsi} with the kernels $\sv_j(\hat\bx, x)$ defined in 
\eqref{eq:psij}. 

We assume that the weights $a = a(x), x\in\R^3,$ and 
$b = b(\hat\bx)$, $\hat\bx\in \R^{3N-3},$ 
are such that $b\in\plainL2_{\textup{\tiny loc}}(\R^{3N-3})$ and  
$a\in\plainL{r}_{\textup{\tiny loc}}(\R^3)$ with some $r >2$. 
As before, denote $\CC_n = (0, 1)^3 + n$,\ $n\in\mathbb Z^3$. 
Let $\varkappa>0$ be the constant in the exponential bound \eqref{eq:exp} and \eqref{eq:FS}. 
We also suppose that $a$ and $b$ are such that both 
\begin{align*}
R_\varkappa(a) =  \sum_{n\in\mathbb Z^3} e^{-  \varkappa|n|_\scale}
\|a\|_{\plainL{r}(\CC_n)} 
\end{align*}
and 
\begin{align*}
M_\varkappa(b) = \biggl[\int_{\R^{3N-3}} 
|b(\hat\bx)|^2 e^{-2\varkappa|\hat\bx|_\scale} 
d\hat\bx\biggr]^{\frac{1}{2}},
\end{align*}
are finite.  
Our objective is to prove the following theorem. 

\begin{thm}\label{thm:psifull}
Let $b\in\plainL2_{\textup{\tiny loc}}(\R^{3N-3})$ and 
$a\in\plainL{r}_{\textup{\tiny loc}}(\R^3)$ with some $r >2$. 
Then $b\, \boldsymbol\SV\, a \in\BS_{1, \infty}$, $j=  1, 2, \dots, N$, and  
\begin{align}\label{eq:vfullest}
\|b\, \boldsymbol\SV\,a\|_{1, \infty}\lesssim M_\varkappa(b) R_\varkappa(a).
\end{align}
\end{thm}

\begin{proof} 
By \eqref{eq:blockvec} it suffices to obtain bounds of the same form 
for the operator $b\,\iop(\sv_j)\,a$ for each fixed $j = 1, 2, \dots, N$. 
By \eqref{eq:FSj}, the kernels $\sv_j$ for all $(\hat\bx, x)\in\SPi_N$ satisfy the bound
\begin{align*}
|\nabla_x^m \sv_j(\hat\bx, x)|\lesssim \big(1+\dc(\hat\bx, x)^{-m}\big)\, 
e^{-\varkappa {|\bx|_\scale}},\quad m = 0, 1, 2, \dots,
\end{align*}
where $\dc(\hat\bx, x)$ is defined in \eqref{eq:dist}.

Thus the kernel $\sv_j(\hat\bx, x)$ satisfies 
the condition \eqref{eq:kernelexp} with 
$t = \hat\bx$, 
$A(\hat\bx, x) = e^{-\varkappa ({|\hat\bx|_\scale + |x|_\scale})}$,
$d = 3$, $n = 3N-3$, $\a = 0$ and $z_k(\hat\bx) =  x_k, k = 0, 1, \dots, N-1$. Furthermore, 
since $r >2$, the weights $a$ and $b$ are as required in Theorem \ref{thm:BSspace}. Thus 
$b\, \iop(\sv_j)\, a\in \BS_{q, \infty}$ where $1/q = 1 + \a/3 = 1$. As 
\begin{align*}
\max_{x\in\CC_n} e^{-\varkappa ({|\hat\bx|_\scale + |x|_\scale})}\le e^{3\varkappa}
e^{-\varkappa |n|_\scale}e^{-\varkappa {|\hat\bx|_\scale}},
\end{align*}
the bound \eqref{eq:vfullest} follows directly from \eqref{eq:piv}. 
\end{proof}

\subsection{Proof of Lemma \ref{lem:central}}\label{subsect:trim} 
The proof of Lemma \ref{lem:central} is divided in several steps each of which is justified 
using either Corollary \ref{cor:zero} or Corollary \ref{cor:zero1}.  
  
The first stage is described in the next lemma. Recall that the cut-offs $Y_\d, Q_R$ and $K_R$  
are defined in \eqref{eq:ydel} and \eqref{eq:qr}.

\begin{lem}
The following relations hold: 
\begin{align}\label{eq:asymp1}
\SfG_1(\boldsymbol\SV) = \lim\limits_{\substack{\d\to 0\\ R\to\infty}} 
\SfG_1(Q_RY_\d \boldsymbol\SV K_R),\quad 
\sg_1(\boldsymbol\SV) 
= \lim\limits_{\substack{\d\to 0\\ R\to\infty}} \sg_1(Q_RY_\d \boldsymbol\SV K_R),
\end{align}
where the limits on the right-hand side exist. 
\end{lem}

\begin{proof} First we check that 
\begin{align}\label{eq:psidelR}
\begin{cases}
\lim\limits_{\d\to 0} 
\|(I- Y_\d)\boldsymbol\SV\|_{1, \infty} = 0,\\[0.3cm] 
\lim\limits_{R\to \infty} \|(I- Q_R)\boldsymbol\SV\|_{1, \infty} = 0,\
\lim\limits_{R\to \infty} \|\boldsymbol\SV(I- K_R)\|_{1, \infty} = 0. 
\end{cases}
\end{align}
Consider $(I-Y_\d)\boldsymbol\SV$. 
By virtue of \eqref{eq:comply}, it follows 
from  \eqref{eq:vfullest} that  
\begin{align*}
\|((1- Y_\d)\boldsymbol\SV\|_{1, \infty}
\lesssim &\ 
M_\varkappa(1-Y_\d) R_\varkappa(1)\\
\lesssim &\ \sum_{0\le l < s\le N-1} 
\bigg[
\int \t\big(|x_l-x_s|(4\d)^{-1}\big)^2 e^{-2\varkappa |\hat\bx|_{\scalel{1}} } d\hat\bx
\bigg]^{1/2}\lesssim \d^{3/2}\to 0,\ \d\to 0,
\end{align*}
and hence the first relation in \eqref{eq:psidelR} holds. 

In a similar way one estimates $(I-Q_R)\boldsymbol\SV$ and 
$\boldsymbol\SV (I-K_R)$. 
Estimate, for example, the second of these operators. 
Since $\z = 1- \t$ (see \eqref{eq:sco}), it follows from \eqref{eq:vfullest} that 
with arbitrary $r >2$,
\begin{align*}
\|\boldsymbol\SV(1-K_R)\|_{1, \infty} 
\lesssim M_\varkappa(1) R_\varkappa(1-K_R) \lesssim &\ \sum_{n\in\mathbb Z^3} e^{-\varkappa|n|_{\scalel{1}}}
\biggl(\int_{\CC_n} |\z(|x| R^{-1})|^r\, dx\biggr)^{\frac{1}{r}}  \\[0.2cm]
\lesssim &\ \sum_{\substack{n\in\mathbb Z^3\\
|n|\ge R/3}} e^{-\varkappa|n|_{\scalel{1}}}
%
%
\to 0, \quad R\to\infty,
\end{align*}
whence the third equality in \eqref{eq:psidelR}. 

Now, represent $\Psi$ in the form 
\begin{align*}
\boldsymbol\SV = Q_R Y_\d\boldsymbol\SV K_R + (I-Q_R)\boldsymbol\SV + Q_R(1-Y_\d)\boldsymbol\SV + 
Q_R Y_\d \boldsymbol\SV (I-K_R).
\end{align*}
According to \eqref{eq:triangle},
\begin{align*}
\|\boldsymbol\SV - Q_R Y_d\boldsymbol\SV K_R\|_{1, \infty}^{\frac{\scalet{1}}{\scalet{2}}} 
\le &\ \|(I-Q_R)\boldsymbol\SV\|_{1, \infty}^{\frac{\scalet{1}}{\scalet{2}}} 
+ \|Q_R(I-Y_\d)\boldsymbol\SV\|_{1, \infty}^{\frac{\scalet{1}}{\scalet{2}}} + 
\|Q_R Y_\d \boldsymbol\SV(I-K_R)\|_{1, \infty}^{\frac{\scalet{1}}{\scalet{2}}}\\[0.2cm]
\le &\ \|(I-Q_R)\boldsymbol\SV\|_{1, \infty}^{\frac{\scalet{1}}{\scalet{2}}}
+ \|(I-Y_\d)\boldsymbol\SV\|_{1, \infty}^{\frac{\scalet{1}}{\scalet{2}}} + 
\|\boldsymbol\SV(I-K_R)\|_{1, \infty}^{\frac{\scalet{1}}{\scalet{2}}}.
\end{align*}
By virtue of \eqref{eq:psidelR} the right-hand side tends to 
zero as $\d\to 0, R\to\infty$.
By \eqref{eq:gnorm} 
Corollary \ref{cor:zero1} ensures \eqref{eq:asymp1}. 
\end{proof}

As the next step towards the kernel \eqref{eq:upsilon} we partition 
the kernel  
\begin{align}\label{eq:trim}
Q_R(\hat\bx) Y_\d(\hat\bx) \SV(\hat\bx, x) K_R(x)
\end{align}
of the operator $Q_R Y_\d\boldsymbol\SV K_R$  
using the cut-offs $\t\big(|x-x_k|\varepsilon^{-1}\big)$, 
$k = 0, 1, \dots, N-1$, under the assumption that $\varepsilon<\d$. 
In view of \eqref{eq:partun} 
the $j$'th component of the kernel \eqref{eq:trim} can be represented as follows:
\begin{align}\label{eq:split}
Q_R(\hat\bx) Y_\d(\hat\bx) \sv_j(\hat\bx, x) K_R(x)
= \sum_{k=0}^{N-1}\phi_{j,k}[\d, R, \varepsilon](\hat\bx, x) 
+ \rho_j[\d, R, \varepsilon](\hat\bx, x)
\end{align}
with 
\begin{align*}
\phi_{j,k}[\d, R, \varepsilon](\hat\bx, x)
= &\ Q_R(\hat\bx) Y_\d(\hat\bx)\t\big(|x-~x_k|\varepsilon^{-1}\big) \sv_j(\hat\bx, x) K_R(x),\quad 
k = 0, 1, \dots, N-1,\\[0.2cm]
\rho_j[\d, R, \varepsilon](\hat\bx, x)
= & \  Q_R(\hat\bx) 
Y_\d(\hat\bx) \prod_{k=0}^{N-1} \z\big(|x-x_k|\varepsilon^{-1}\big)  \sv_j(\hat\bx, x) K_R(x).
\end{align*} 
First we show that the kernels $\rho_j[\d, R, \varepsilon]$ and $\phi_{j,0}[\d, R, \varepsilon]$ 
give negligible contributions to the asymptotics. 

\begin{lem} For each $\d>0, R>0$ and $\varepsilon<\d$ one has 
\begin{align}\label{eq:rho}
\SfG_1\big(\iop\big(\rho_j[\d, R, \varepsilon]\big)\big) = 0,\ j = 1, 2, \dots, N.
\end{align}
\end{lem}

\begin{proof} 
By the definitions \eqref{eq:ydel} and \eqref{eq:sco1}, 
the support of the kernel 
$\rho_j[\d, R, \varepsilon]$ belongs to the bounded domain 
\begin{align*}
 {\bigcap_{0\le l < s \le N}  \SfS_{l, s}(\varepsilon/2)\cap\, B_R}. 
\end{align*}
The function $\sv_j$ is real-analytic on this domain and it is uniformly bounded 
together with all its derivatives, so that $\rho_j[\d, R, \varepsilon]\in \plainC\infty_0(\R^{3N})$.
By Proposition \ref{prop:BS}, $\iop(\rho_j[\d, \R, \varepsilon])\in\BS_{p, \infty}$ 
for all $p >0$, whence \eqref{eq:rho}, as claimed. 
\end{proof}
 
\begin{lem} For each $\d>0, R>0$ one has 
\begin{align}\label{eq:phi0}
\lim_{\varepsilon\to 0}\SfG_1\big(\iop\big(\phi_{j,0}[\d, R, \varepsilon]\big)\big) = 0,\ 
j = 1, 2, \dots, N.
\end{align}
\end{lem}

\begin{proof}
As $x_0 = 0$ by definition, the kernel $\phi_{j,0}[\d, R, \varepsilon]$ has the form 
\begin{align*}
\phi_{j,0}[\d, R, \varepsilon](\hat\bx, x)
= Q_R(\hat\bx) Y_\d(\hat\bx)\sv_j(\hat\bx, x) \t\big(|x|\varepsilon^{-1}\big)  K_R(x). 
\end{align*}
Estimating $Q_R Y_\d\le 1$, $ K_R\le 1$, one sees that 
the singular values of $\iop(\phi_{j,0}[\d, R, \varepsilon])$ do not exceed those of the operator 
$\iop(\sv_j) a$ with the weight $a(x) = \t(|x|\varepsilon^{-1})$. 
By \eqref{eq:vfullest}, for arbitrary $r >2$ we have 
\begin{align*} 
\|\iop(\sv_j)a)\|_{1, \infty}
\lesssim M_\varkappa(1) R_\varkappa(a)
\lesssim \bigg(\int_{\R^3} \t\big(|x|\varepsilon^{-1} \big)^r dx\bigg)^{\frac{1}{r}}
\lesssim \varepsilon^{\frac{3}{r}} \to 0,\ \varepsilon\to 0.
\end{align*}
This implies \eqref{eq:phi0}.
\end{proof}

Let us now analyze the remaining components of \eqref{eq:split}.

\begin{cor}
Denote by $\boldsymbol\a[\d, R, \varepsilon]  = \{\a_j[\d, R, \varepsilon]\}_{j=1}^N$ 
the vector-valued kernel with the components  
\begin{align*}
\a_j[\d, R, \varepsilon](\hat\bx, x) = \sum_{k=1}^{N-1}\phi_{j,k}[\d, R, \varepsilon](\hat\bx, x).
\end{align*}
Then for all $\d >0$ and $R>0$, we have  
\begin{align}\label{eq:asymp2}
\begin{cases}
\SfG_1(Q_RY_\d \boldsymbol\SV K_R) = &\ 
\lim\limits_{\varepsilon\to 0} 
\SfG_1(\iop(\boldsymbol\a[\d, R, \varepsilon])),
\\[0.2cm]
\sg_1(Q_RY_\d \boldsymbol\SV K_R) = &\ 
\lim\limits_{\varepsilon\to 0} 
\sg_1(\iop(\boldsymbol\a[\d, R, \varepsilon])),
\end{cases}
\end{align}
where the limits on the right-hand side exist.
\end{cor}

\begin{proof}
By \eqref{eq:split}, the kernel $Q_R Y_\d \sv_j K_R$ has the form 
\begin{align*}
\a_j[\d, R, \varepsilon] + \phi_{j, 0}[\d, R, \varepsilon] + \rho_j[\d, R, \varepsilon].
\end{align*}
By virtue of \eqref{eq:trianglep} and \eqref{eq:rho}, \eqref{eq:phi0}, 
for each $j = 1, 2, \dots, N$ 
we have 
\begin{align*}
\lim_{\varepsilon\to 0}
\SfG_1\big(\iop\big(\phi_{j, 0}[\d, R, \varepsilon] + \rho_j[\d, R, \varepsilon]\big)\big) = 0.
\end{align*}
Now \eqref{eq:asymp2} follows from Corollary \ref{cor:zero1}.
\end{proof}

\begin{proof}[Completion of the proof of Lemma \ref{lem:central}]
According to 
\eqref{eq:cutoff}, under the condition $\varepsilon < \varepsilon_0(\d, R)$, the 
support of each kernel
\begin{align*}
\phi_{j, k}[\d, R, \varepsilon], \quad j= 1, 2, \dots, N,\quad k = 1, 2, \dots, N-1,
\end{align*}  
belongs to $\tilde\Om_k(\d, R, \varepsilon)$, see \eqref{eq:omj} 
for the definition. Therefore one can use the representation \eqref{eq:trepr} for the function 
$\sv_j$: 
\begin{align*}
\a_j[\d, R, \varepsilon](\hat\bx, x) 
= \sum_{k=1}^{N-1}&\ \phi_{j,k}[\d, R, \varepsilon](\hat\bx, x)
= \sum_{k=1}^{N-1} Q_R(\hat\bx) Y_\d(\hat\bx)\t\big(|x-x_k|\varepsilon^{-1}\big) 
\nabla_x\tilde\xi_{j,k}(\hat\bx, x) K_R(x)\\[0.2cm] 
&\ \quad + \sum_{k=1}^{N-1} Q_R(\hat\bx) Y_\d(\hat\bx)\t\big(|x-x_k|\varepsilon^{-1}\big) 
|x_k-x|\nabla_x\tilde\eta_{j,k}(\hat\bx, x) K_R(x)\\[0.2cm]
&\ \quad + \sum_{k=1}^{N-1} Q_R(\hat\bx) Y_\d(\hat\bx)\t\big(|x-x_k|\varepsilon^{-1}\big) 
\frac{x-x_k}{|x-x_k|}\tilde\eta_{j,k}(\hat\bx, x) K_R(x).
\end{align*}
Each term in the first sum on the right-hand side belongs to $\plainC\infty_0(\R^{3N})$.
 Thus, by Proposition \ref{prop:BS}, the functional 
$\SfG_p$ for the associated operator equals zero for all $p >0$, and in particular, for $p=1$.
Each term in the second sum contains the function $\Phi(x) = |x|$, which is homogeneous of order $1$. 
Thus, by Proposition \ref{prop:gradp} the functional $\SfG_1$ for the associated operator equals zero. 
The third sum coincides with the kernel $\Upsilon_j[\d, R, \varepsilon](\hat\bx, x)$, 
defined in \eqref{eq:upsilon}. Therefore, by Corollary \ref{cor:zero}, 
\begin{align}\label{eq:asymp3} 
\begin{cases}
\SfG_1(\iop(\boldsymbol\a[\d, R, \varepsilon])) = 
\SfG_1(\iop(\boldsymbol\Upsilon[\d, R, \varepsilon])),\\[0.2cm]
\sg_1(\iop(\boldsymbol\a[\d, R, \varepsilon])) = 
\sg_1(\iop(\boldsymbol\Upsilon[\d, R, \varepsilon])),
\end{cases}
\end{align}
for each $\d>0, R>0$ and $\varepsilon<\varepsilon_0(\d, R)$. 
Putting together \eqref{eq:asymp1}, \eqref{eq:asymp2} and \eqref{eq:asymp3}  
and using Corollary \ref{cor:zero1}, 
we conclude the proof of Lemma \ref{lem:central}.
\end{proof}

Recall that the main Theorem \ref{thm:maincompl} was derived from Lemma \ref{lem:central} 
in Subsect. \ref{subsect:proofs}. Thus all proofs are complete.

\section{Appendix}\label{sect:app}

In this appendix we prove an elementary extension result for Sobolev spaces. 
We consider spaces of functions that depend either 
on one variable $x\in\R^d$ or on two variables $(t, x)\in \R^l\times\R^d$. 
Denote $K = \{t\in\R^l: |t| < 1\}$,\ $B = \{x\in\R^d: |x| < 1\}$ and 
$B_0 = B\setminus\{0\}$. In Section \ref{sect:intop} (see Remark \ref{rem:remove}) 
we need only the case of one variable $x\in \R^d$. The following lemma however holds for 
the more general case of two variables $(t, x)\in \R^l\times \R^d$ as well, 
which may come in handy in different 
circumstances.

Our objective is to establish the following elementary extension result. 

\begin{lem} \label{lem:remove}
Let the dimension $l$ be arbitrary, let 
$d\ge 2$, $m\ge 1$, and 
%
%
$p\in [d(d-1)^{-1}, \infty]$. Then 
$\plainW{m, p}(B_0) 
= \plainW{m, p}(B)$ and $\plainW{m, p}\big(K\times B_0\big) 
= \plainW{m, p}\big(K\times B\big)$.
\end{lem}  
  
\begin{proof} 
We prove only the second equality since the first one is obtained in the same way. 
It suffices to show that $\plainW{1, p}\big(K\times B_0\big) 
\subset \plainW{1, p}\big(K\times B\big)$. 
Let $u\in \plainW{1, p}\big(K\times B_0\big)$. 
We show that for all $j\in \mathbb N_0^{d+l}$, $|j|=1$,  and all test 
functions $\phi\in \plainC\infty_0(K\times B)$ the following identity holds:
\begin{align}\label{eq:parts}
\int u\, \p^j\phi\, dt dx = -\int \p^j u \, \phi \, dt dx,  
\end{align}
where by $\p^j u $ we denote the distributional  
derivatives of $u$ on the set $K\times B_0$. 
Let $\t$ be the function as defined in \eqref{eq:sco} and \eqref{eq:sco1}, and introduce 
$\t_\varepsilon(x) = \xi(|x|\varepsilon^{-1})$. Then 
\begin{align*}
\int u\, \p^j\phi\, dt dx = \int u\, \p^j\big(\phi\,(1-\t_\varepsilon)\big)\, dt dx 
+ \int u\, \p^j\big(\phi\,\t_\varepsilon\big)\, dt dx. 
\end{align*}
Since $1-\t_\varepsilon(x) = 0$ for $|x|<\varepsilon/2$, in the first integral 
on the right-hand side we can integrate by parts:
\begin{align}\label{eq:backpar}
\int u\, \p^j\phi\, dtdx = 
&\ -\int \p^j\, u\, \phi\,(1-\t_\varepsilon)\, dt dx  
+ \int u\, \p^j\big(\phi\,\t_\varepsilon\big)\, dt dx\notag\\ 
= &\ -\int \p^j\, u\, \phi\, dt dx
+ \int \p^j\, u\, \phi\,\t_\varepsilon\, dt dx 
+ \int u\, \p^j\big(\phi\,\t_\varepsilon\big)\, dt dx.
\end{align}
Further proof is conducted for $p <\infty$ only. 
The modifications for the case $p = \infty$ are obvious. 
By H\"older's inequality, the second term on the right-hand side of 
\eqref{eq:backpar} estimates as follows:
\begin{align*}
\bigg|\int \p^j\, u\, \phi\,\t_\varepsilon\, dt dx\bigg|
\le \bigg[\int\limits_{|x|< \varepsilon}|\p^j u|^p\, dt dx\bigg]^{\frac{1}{p}} \, 
\bigg[\int\limits_{|x|< \varepsilon}\, dtdx\bigg]^{\frac{1}{q}}, 
\end{align*}
where $q^{-1} = 1- p^{-1}$. Thus the right-hand side does not exceed 
\begin{align}\label{eq:1}
\varepsilon^{\frac{d}{q}}\bigg[\int\limits_{|x|< \varepsilon}|\p^j u|^p\, dt dx\bigg]^{\frac{1}{p}}
\to 0,\quad \textup{as}\quad \varepsilon\to 0.
\end{align} 
Consider the third integral on the right-hand side of \eqref{eq:backpar}. 
By H\"older's inequality 
again, 
\begin{align*}
\bigg|\int u\, \p^j\big(\phi\,\t_\varepsilon\big)\, dtdx\bigg|
\le \bigg[\int\limits_{|x|< \varepsilon}|u|^p\, dt dx\bigg]^{\frac{1}{p}} \,
\max|\p^j (\phi\,\t_\varepsilon)|\, 
\bigg[\int\limits_{|x|< \varepsilon}\, dt dx\bigg]^{\frac{1}{q}}. 
\end{align*}
%
%
As $\max|\p^j (\phi\t_\varepsilon)|\lesssim \varepsilon^{-1}$, and $q = (1-p^{-1})^{-1} \le d$, 
the right-hand side does not exceed 
\begin{align*}
\varepsilon^{\frac{d}{q} - 1}\bigg[\int\limits_{|x|< \varepsilon}|u|^p\, dtdx\bigg]^{\frac{1}{p}}
\to 0,\quad \textup{as}\quad \varepsilon\to 0.
\end{align*}
Together with \eqref{eq:1} this implies that the last two terms on the right-hand side 
of \eqref{eq:backpar} tend to zero as $\varepsilon\to 0$. This entails 
\eqref{eq:parts}, as required. This completes the proof of the lemma.
\end{proof}      
      
\vskip 0.5cm

  \textbf{Acknowledgments.} 
The author is grateful to J. Cioslowski    
for comments and stimulating discussions, and to M. Lewin for pointing out the asymptotic 
problem for the kinetic energy density operator. 
Thanks also go to A. Nazarov 
and A. Tyulenev for bringing to the author's attention the book 
\cite{Burenkov1998}, 
and to D. Edmunds, V. Kozlov, V. Maz'ya, G. Rozenblum, D. Vassiliev 
for their advice on Sobolev spaces.  

The author was supported by the EPSRC grant EP/P024793/1.

\end{document}